\documentclass{article}

 \usepackage[main, final]{neurips_2026}
\usepackage[utf8]{inputenc} % allow utf-8 input
\usepackage[T1]{fontenc}    % use 8-bit T1 fonts
\usepackage{hyperref}       % hyperlinks
\usepackage{url}            % simple URL typesetting
\usepackage{booktabs}       % professional-quality tables
\usepackage{amsfonts}       % blackboard math symbols
\usepackage{nicefrac}       % compact symbols for 1/2, etc.
\usepackage{microtype}      % microtypography
\usepackage{xcolor}         % colors

\usepackage{amsmath}
\usepackage{amssymb}
\usepackage{mathtools}
\usepackage{amsthm}

\usepackage{bm}

\usepackage{wrapfig}
\usepackage{color}
\usepackage{colortbl}
\usepackage{booktabs}
\usepackage[dvipsnames,table]{xcolor}

\usepackage{makecell}

\usepackage{array}

\newcommand{\thickhline}{\specialrule{\heavyrulewidth}{0pt}{0pt}}

\definecolor{colPaper}{RGB}{234, 240, 254}   % Light Periwinkle
\definecolor{colModel}{RGB}{252, 235, 238}   % Light Pink
\definecolor{colComp}{RGB}{241, 248, 234}    % Light Lime/Green
\definecolor{colParam}{RGB}{254, 248, 228}   % Light Cream/Yellow
\definecolor{colBound}{RGB}{225, 235, 236}   % Light Gray-Blue

\definecolor{symRed}{HTML}{C62828}      % Deep Crimson
\definecolor{symGreen}{HTML}{2E7D32}    % Forest Green
\definecolor{symWarn}{HTML}{F9A825}     % Muted Amber

\theoremstyle{plain}
\newtheorem{theorem}{Theorem}[section]
\newtheorem{proposition}[theorem]{Proposition}

\theoremstyle{definition}

\newtheorem{assumption}[theorem]{Assumption}
\theoremstyle{remark}

\usepackage{algorithm}
\usepackage{algpseudocode}

\usepackage{subcaption}

\usepackage[most]{tcolorbox}
\definecolor{apricotBG}{HTML}{FFE9DC}
\definecolor{terracottaBorder}{HTML}{A0522D}

\newtcolorbox{myremark}[1]{
  enhanced,
  colback=apricotBG,
  colframe=terracottaBorder,
  coltitle=white,
  fonttitle=\bfseries,
  title={Remark: #1},
  attach boxed title to top left={yshift=-2mm, xshift=3mm},
  boxed title style={colback=terracottaBorder},
  arc=3pt,
  boxrule=1pt,
  drop fuzzy shadow
}

\title{Generalized Neural Operator for Parametric and Boundary-Value Problems}

\author{%
  Ruoyan Li\\
  University of California, Los Angeles \\
  \And
  Yizhou Sun \\
  University of California, Los Angeles \\
  \AND
  Wei Wang \\
  University of California, Los Angeles \\
}

\begin{document}

\maketitle

\begin{abstract}
  Physical systems are fully characterized by Partial Differential Equations (PDEs) alongside their specific parameters and boundary conditions.
  Thus, developing foundational physical simulators requires robust generalization across diverse PDE parameters and boundary conditions.
  However, existing solvers face a fundamental trilemma, struggling to simultaneously achieve computational efficiency, mathematical rigor, and physics-agnostic deployment.
  Traditional numerical methods provide mathematical rigor but are computationally expensive.
  Physics-Informed Neural Networks (PINNs) maintain this rigor and are computationally efficient, yet they remain instance-specific and require retraining for new conditions.
  Purely data-driven Neural Operators offer efficiency but sacrifice explicit physical grounding, leading to ill-posed formulations, dataset biases, and generalization failures.
  Furthermore, the massive scale of emerging foundational operators has severely degraded their inference speeds, making them computationally uncompetitive with traditional numerical solvers.
  To address these bottlenecks, we propose a \textit{Generalized Neural Operator} that bridges these paradigms by restoring mathematical rigor through explicit physical conditioning.
  By formalizing the classical conditions for well-posedness within neural operators, our framework demonstrates the theoretical benefits of explicitly conditioning on PDE parameters and boundary conditions.
  To implement this synthesis without compromising computational speed, we introduce three novel architectural components: a parameter-gated mixture of kernels for efficient parameter generalization, a generalized boundary transfer operator that projects arbitrary boundary constraints into a unified latent Dirichlet representation, and a specialized training objective to ensure stability.
  Extensive experiments demonstrate that our theoretically grounded approach achieves superior generalization across heterogeneous physical regimes while maintaining strict inference efficiency comparable to conventional numerical baselines.

\end{abstract}

\section{Introduction}
\label{sec:intro}
Partial Differential Equations (PDEs) constitute the standard mathematical framework for modeling physical phenomena. To fully characterize a specific physical system, these equations are formulated as parametric boundary value problems, in which the combination of equation parameters and boundary conditions fully defines the underlying physics. For example, the heat equation serves as a canonical model for understanding diffusion processes, describing how a quantity such as temperature evolves over time to reach equilibrium. $\dfrac{\partial u(\bm{x}, t)}{\partial t} = \alpha \nabla^2 u(\bm{x}, t)$, where $u(\bm{x}, t)$ represents the temperature at spatial position $\bm{x}$ and time $t$. Central to this physical model is the thermal diffusivity parameter, $\alpha$, which dictates the intrinsic rate at which heat propagates through a specific medium. A high diffusivity characterizes materials like copper that rapidly conduct thermal energy, while low diffusivity defines insulators that retard this flow. The physical behavior of the system is also constrained by its interactions with the environment, defined mathematically as boundary conditions. Dirichlet condition, $u(\bm{x}, t) = f(\bm{x}, t)$ for $\bm{x} \in \partial\Omega$, is a scenario where the boundary temperature is strictly controlled, such as an object submerged in a thermal bath. Neumann condition, $\frac{\partial u}{\partial \bm{n}}(\bm{x}, t) = g(\bm{x}, t)$ for $\bm{x} \in \partial\Omega$, with $\mathbf{n}$ denoting the outward unit normal vector, represents the control of heat flux at the boundary, where a homogeneous value, $g(\bm{x}, t)=0$, implies a perfectly insulated surface that allows no energy escape. Periodic conditions simulate a ring-like or infinite geometry where heat leaving one side re-enters the other. Thus, because real-world applications involve diverse materials and varying environmental interactions, learning a foundational physical neural model requires generalization across PDE parameters and boundary conditions.

\begin{wrapfigure}{r}{0.5\textwidth}
  \centering
  \includegraphics[width=0.45\textwidth]{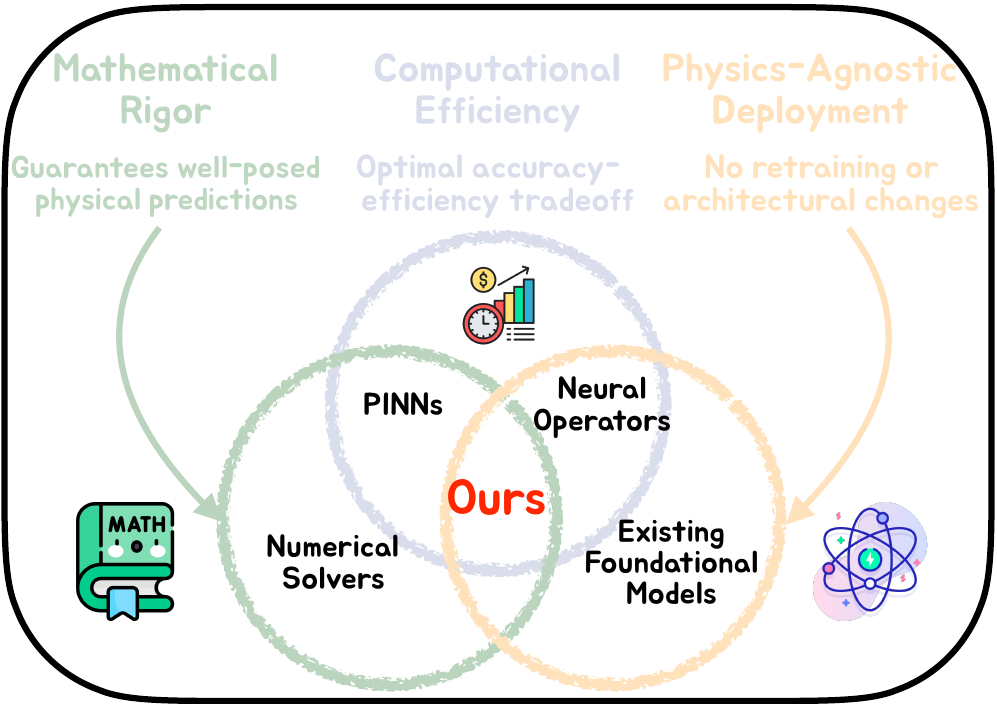}
  \caption{Trilemma of existing PDE solvers. Traditional numerical solvers guarantee mathematical rigor but are computationally expensive and demand instance-specific design. PINNs maintain this rigor and offer efficient inference, yet still require instance-specific retraining. Neural operators achieve favorable accuracy-efficiency tradeoffs and physics-agnostic deployment, but often yield ill-posed predictions. Existing foundational models achieve physics-agnostic deployment but suffer from computational inefficiency due to their massive size. Our proposed method unifies the core strengths of all three paradigms.}
  \label{fig:intro}
\end{wrapfigure}

Existing PDE solvers for these parametric boundary value problems struggle to balance three competing objectives: computational efficiency, mathematical rigor, and physics-agnostic deployment. Traditional numerical methods, such as finite difference and finite element solvers, prioritize mathematical rigor by computing state evolutions with strict, explicit adherence to governing equations. However, they are extremely computationally expensive and require entirely different, instance-specific algorithmic designs. For example, a seemingly simple shift from Dirichlet to periodic boundary conditions cannot be handled via a simple switch. Instead, it demands completely replacing a finite difference scheme with a fundamentally different algorithm, such as a pseudo-spectral method. Physics-Informed Neural Networks (PINNs)~\citep{RAISSI2019686} maintain mathematical rigor by embedding explicit equations into the loss function, but they remain instance-specific, requiring costly retraining for every new condition and thereby sacrificing agnostic deployment. Purely data-driven Neural Operators~\citep{kovachki2023neuraloperator} achieve the "learn once, evaluate instantly" advantage by learning infinite-dimensional mappings from large-scale datasets and demonstrate superior accuracy-efficiency tradeoffs~\citep{wang2025fdbenchmodularfairbenchmark}. However, by treating the physics implicitly and attempting to deduce governing laws primarily from historical frames, they pose an ill-posed inverse problem, sacrificing mathematical rigor and leaving them vulnerable to generalization failure.

Current attempts to build foundational physical models struggle precisely because they fail to resolve this trilemma. In pursuit of physics-agnostic deployment, most lean heavily into the data-driven operator paradigm but abandon the explicit physical grounding that makes numerical solvers and PINNs robust. By omitting explicit boundary information and parameter conditioning, these models often fall back on dataset biases, producing averaged, non-physical predictions. Furthermore, to compensate for this lack of explicit rigor, foundational models are scaled to massive sizes, which ironically destroys their computational efficiency. Given that recent studies~\citep{wang2025fdbenchmodularfairbenchmark, McGreivy_2024} show even lightweight operators (e.g., the 0.5M parameter Fourier Neural Operator~\citep{li2021fourier}) offer only marginal speedups over highly optimized numerical solvers, the practical utility of deploying massive, multi-million parameter models to guess implicit physics is highly questionable.

To definitively bridge the gap between computational efficiency, mathematical rigor, and condition-agnostic deployment, we propose a \textit{Generalized Neural Operator}. We formalize the well-established classical conditions for well-posedness within the context of neural operators, demonstrating that restoring mathematical rigor through explicit physical inputs is the key to true generalization. To achieve this synthesis without sacrificing efficiency, our architecture introduces a parameter-gated mixture of kernels, dynamically routing computation based on explicit physical regimes rather than activating a monolithic network. This ensures physics-agnostic deployment while maintaining runtime efficiency strictly comparable to traditional numerical solvers. Furthermore, we propose a generalized boundary transfer operator that projects arbitrary, explicitly provided boundary conditions into a unified latent Dirichlet representation, allowing the core network to maintain rigor while seamlessly generalizing across configurations. Finally, we introduce a specialized training objective to ensure robust performance across diverse physical regimes.

Our contributions are as follows: \textbf{(i) Paradigm Synthesis \& Problem Identification:} We identify the fundamental struggle in current foundational PDE models and propose a framework that resolves this trilemma. \textbf{(ii) Theoretical Formalization:} We formalize the well-posedness of foundational neural solvers, demonstrating theoretically why the structural inclusion of PDE parameters and boundary conditions is beneficial. \textbf{(iii) Practical Solution:} We introduce a parameter-gated mixture of kernels, a generalized boundary transfer operator, and a robust optimization objective to achieve generalization across diverse physical regimes. \textbf{(iv) Experimental Validation:} We validate the superior generalization performance of our proposed model through extensive empirical benchmarks.

\section{Related Work}
\label{sec:related}
Early research on neural operators primarily focuses on learning PDEs with single parameters and fixed boundary conditions. This body of work encompasses Eulerian neural simulators~\citep{li2021fourier, li2023physicsinformedneuraloperatorlearning, pfaff2021learning, brandstetter2022message, rahman2023uno, wu2024Transolver, luo2025transolver, li2023geometryinformedneuraloperatorlargescale, tran2023factorized, wu2023LSM}, which operate on regular grids or irregular meshes, and Lagrangian neural simulators~\citep{sanchezgonzalez2020learning, toshev2024lagrangebench, Prantl2022Conserving, toshev2024neuralsphimprovedneural}, which discretize the domain into particles.

Recent initiatives have sought to develop foundational models that extend the capabilities of neural operators in several directions. Research such as \citet{alesiani2022hyperfno, takamoto2023learningneuralpdesolvers} focuses on generalizing operators across PDE parameters, while \citet{boudec2025learning} introduces a physics-guided iterative training algorithm to handle diverse parameters and boundary conditions. Unified frameworks have also emerged, including the work by \citet{alkin2024upt} for joint Eulerian and Lagrangian simulations and \citet{chen2025omniarch} for multiscale pretraining across 1D, 2D, and 3D domains. \citet{ye2024pdeformer} utilizes computational graphs to represent various PDEs, and \citet{shen2024ups} embeds systems into shared spaces to enable generalization across dimensions and resolutions. % \YS{limitation of these work?}

\citet{saad2023guiding} investigates architectural designs for enforcing exact boundary conditions, and \citet{Yang_2023, yang2024pdegeneralizationincontextoperator} employs in-context learning to infer operators dynamically using example pairs as prompts. Furthermore, transformer-based multiphysics solvers~\citep{hao2024dpot, wang2025mixtureofexpertsoperatortransformerlargescale, mccabe2024multiple, subramanian2023towards, sun2025foundationmodelpartialdifferential, herde2024poseidon, zhou2025unisolver} conduct pretraining on large-scale datasets containing numerous PDEs and conditions. Some of these works rely on fine-tuning to generalize well to new physics. We approach the problem from a numerical PDE perspective by learning models that generalize to new parameters and boundary conditions without any instance-specific fine-tuning. % \YS{do we need fine-tuning?}

The landscape of scientific machine learning encompasses a diverse array of applications that extend far beyond forward modeling.  Significant research has been dedicated to tasks such as inferring dense fluid fields from sparse observations~\citep{li2025flow, zhong2023sparse, mo2024reconstructingunsteadyflowssparse, yadav2025rfpinns, jing2024airflow, He_2022_flow} and enhancing the performance of classical numerical schemes with neural networks~\citep{list2022learned, sun2023neural, greenfeld2019learning, sappl2019deep}. Furthermore, neural networks play a crucial role in engineering workflows through aerodynamic shape optimization~\citep{elrefaie2025drivaernet, NEURIPS2024_013cf29a} and inverse design methodologies~\citep{behrmann2019invertible, teng2019invertible, kruse2021benchmarking}. While these contributions fall broadly within the same context of scientific machine learning as our work, they address objectives distinct from the primary focus of this study.

% While our research aligns with foundational PDE solvers, it makes distinct contributions by explicitly incorporating PDE parameters and boundary conditions as inputs.cWe focus on generalizing within specific PDE families, rather than disparate physical systems, to ensure superior efficiency compared to conventional numerical solvers. Our work is most closely related to \citet{zhou2025unisolver}, which also employs explicit incorporation. However, we diverge by studying a more efficient architecture rather than a large-scale model with millions of parameters, which ensures efficiency against numerical solvers. We provide theoretical justification for the necessity of explicit incorporation of PDE parameters and boundary conditions to ensure well-posedness. Crucially, we generalize beyond the scalar parameters and local boundary conditions found in \citet{zhou2025unisolver}. Our approach handles spatially dependent coefficients, such as the velocity field in the Advection equation (Section~\ref{sec:exp}), and employs a \textit{Generalized Boundary Transfer Operator} to support arbitrary boundary constraints, including integral and non-local conditions.

\begin{figure*}[t]
    \centering
    \includegraphics[width=0.80\textwidth]{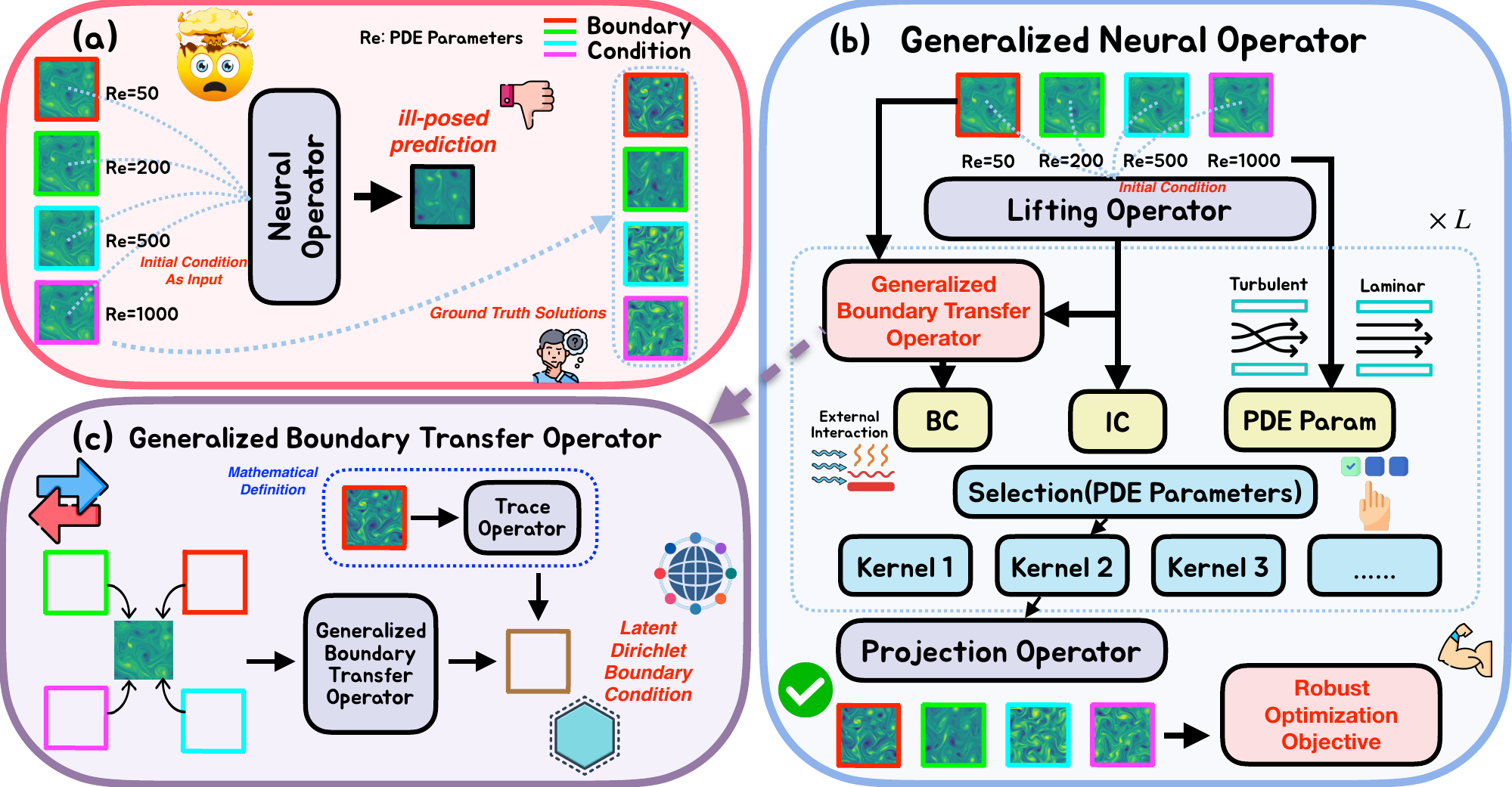}
    \caption{
    \textbf{(a)} Lack of physical inputs renders the problem ill-posed, causing the model to converge toward an averaged solution.
    \textbf{(b)} We propose a \textit{Generalized Neural Operator}, which incorporates a parameter-gated mixture of kernels and a robust optimization training objective to improve generalization.
    \textbf{(c)} To enhance boundary generalization, we introduce a generalized boundary transfer operator that maps arbitrary boundary conditions into a latent Dirichlet boundary state. This unified modality significantly aids model generalization.
    }
    \label{fig:framework}
\end{figure*}

\section{Problem Statement}
\label{sec:problem}
We focus on generalizing within specific PDE families, rather than disparate physical systems, to ensure efficiency compared to conventional numerical solvers. We consider the \textbf{Parametric Boundary Value Problem}. Let $\Omega \subseteq \mathbb{R}^d$ be a bounded domain with boundary $\partial \Omega$. We consider a family of first-order in time partial differential equations with parameters $\bm{\theta} \in \Theta \subseteq \mathbb{R}^p$:
\begin{align} \label{eq:problem}
\begin{cases} 
\dfrac{\partial \bm{u}}{\partial t} + \mathcal{L}_{\bm{\theta}}[\bm{u}] = 0 & (\bm{x}, t) \in \Omega \times (0, T] \\ 
\mathcal{B}[\bm{u}] = 0 &  \\
\bm{u}(\bm{x}, 0) = \bm{u}_0(\bm{x}) & \bm{x} \in \Omega 
\end{cases}
\end{align}

where $\mathcal{L}_{\bm{\theta}}$ is the parameterized differential operator, $\bm{u}_0$ is the initial condition, and $\bm{u} (\bm{x}, t)$ is the solution of interest. $\mathcal{B}$ is the boundary operator. This operator specifies the type of the boundary constraint enforced and the associated boundary data, such as the specific Dirichlet values or the magnitude of the Neumann flux. We adopt the notation where square brackets $[\cdot]$ denote an operator on a function, while parentheses $(\cdot)$ denote the evaluation of a function at specific coordinates. % \YS{notation wise, what's the difference between [] and ()?}

Let $\mathcal{Y}$ denote the topological vector space representing the solution trajectory. Let $\mathcal{X}_{\text{init}}$ denote the topological vector space representing the initial conditions. We define the solution operator $\mathcal{S}$ with respect to a fixed configuration of the PDE structure. Specifically, for a fixed parameter $\bm{\theta}$ and boundary $\mathcal{B}$, the operator $\mathcal{S}_{\bm{\theta}, \mathcal{B}}: \mathcal{X}_{\text{init}} \rightarrow \mathcal{Y}$ maps the initial condition to the solution trajectory:
\begin{align}
    \bm{u} = \mathcal{S}_{\bm{\theta}, \mathcal{B}}(\bm{u}_0)
\end{align}
where $\bm{u}$ satisfies the system in Equation~\ref{eq:problem}. 

Assume the structural form of the differential operator $\mathcal{L}$ is fixed. We consider variations in the PDE parameters $\bm{\theta} \in \Theta$,  boundary operator $\mathcal{B} \in \mathbb{B}$, where $\mathbb{B}$ is the set of admissible boundary operators. Let $\mathfrak{S} = \{ \mathcal{S}_{\bm{\theta}, \mathcal{B}} \mid \bm{\theta} \in \Theta , \mathcal{B} \in \mathbb{B}  \}$ denote the set of all solution operators arising from admissible combinations of these varying components. Our goal is to learn a \textbf{Generalized Neural Operator} $\mathcal{M}_{\phi}$, parameterized by $\phi$, that approximates the family of operators in $\mathfrak{S}$.

\section{Method}
We first formalize the theoretical conditions required for operator well-posedness, demonstrating that explicit input of PDE parameters and boundary conditions is beneficial. Building on this mathematical foundation, we introduce a parameter-gated mixture of kernels, a generalized boundary transfer operator, and a robust optimization training objective.

\subsection{Theoretical Formalization}

\begin{proposition}
\label{prop:well}
    Let $\mathfrak{S} = \{ \mathcal{S}_{\bm{\theta}, \mathcal{B}} \mid \bm{\theta} \in \Theta , \mathcal{B} \in \mathbb{B} \}$ be the family of solution operators defined in Equation~\ref{eq:problem}.  Then, a single-valued operator $\mathcal{M}_{\phi}: \mathcal{X}_{\mathrm{init}} \to \mathcal{Y}$ defined solely on the initial condition space cannot universally represent the family $\mathfrak{S}$. For the \textit{Generalized Neural Operator} to be well-defined as a function approximating $\mathfrak{S}$, its domain must be augmented to include the parameter spaces $\Theta$ and boundary operator space $\mathbb{B}$.
\end{proposition}
\begin{proof}
    See Appendix~\ref{appendix:proof_well}
\end{proof}
We present the well-posedness of the neural operator in Proposition~\ref{prop:well}. It reveals a fundamental mathematical limitation in purely implicit foundational models. Since identical initial conditions can lead to entirely different trajectories under distinct physical configurations, an operator that maps solely from the initial condition space to the solution space lacks sufficient constraints to isolate a single, physically valid outcome.

\begin{myremark}{Ill-Posedness of Implicit Parameter Inference}
    While recent works~\citep{yin2022leadslearningdynamicalsystems, koupai2024gepsboostinggeneralizationparametric, nzoyem2025neural, koupai2025enma} have demonstrated competitive empirical performance when implicitly inferring PDE parameters from historical state trajectories, this approach fundamentally reframes the forward simulation as a hidden inverse problem. Mathematically, deducing governing parameters solely from observed states is inherently ill-posed. Consequently, reliance on implicit inference introduces an irreducible epistemic uncertainty.
\end{myremark}

\begin{myremark}{Causal Impossibility of Implicit Boundary Inference}
    While static boundaries can be inferred from rich historical context, dynamic or time-dependent boundary conditions present a hard causal barrier. Consider a time-varying Dirichlet boundary condition. It is fundamentally impossible for a neural operator to causally deduce the unobserved, future state of this external influence from the system's past sequence. In such scenarios, implicit inference is strictly impossible.
\end{myremark}

We formulate \textit{Generalized Neural Operator} $\mathcal{M}_{\phi}$ over an augmented domain that incorporates the PDE configuration. We define $\mathcal{M}_{\phi}$ as a mapping from the Cartesian product of the initial condition space, PDE parameter space, and boundary operator space to the solution trajectory space. Mathematically, this is expressed as:
\begin{align}
    \mathcal{M}_{\phi} : \mathcal{X}_{\text{init}} \times \Theta \times \mathbb{B} \to \mathcal{Y}
\end{align}
Under this formulation, the neural operator takes the tuple $(\bm{u}_0, \bm{\theta}, \mathcal{B})$ as input and predicts the corresponding solution trajectory $\hat{\bm{u}}$ that approximates the true solution defined by the operator $\mathcal{S}_{\bm{\theta}, \mathcal{B}}(\bm{u}_0)$. This augmentation ensures that $\mathcal{M}_{\phi}$ is well-defined, capable of distinguishing between varying physical dynamics and boundary constraints. The PDE parameters are embedded directly from their numerical values, while the boundary conditions are encoded using a combination of their type and the specific boundary data.

Our approach readily generalizes to higher-order temporal PDEs by augmenting the input with additional initial snapshots. For example, in Newtonian dynamics governed by second-order time derivatives, the input would consist of both the initial position and the initial velocity.

\subsection{Parameter-Gated Mixture of Kernels}
Building upon the neural operator framework~\citep{kovachki2023neuraloperator}, our architecture consists of three components: (1) a local lifting transformation, (2) a sequence of iterative kernel integration layers, and (3) a local projection to the output space. A central innovation in our approach is a parameter-gated mixture of kernels, where the integral kernels are dynamically selected based on the PDE parameters.
\begin{align}
    \mathcal{M}_\phi = \left( \mathcal{Q} \circ \Phi_L \circ \Phi_{L-1} \circ \cdots \circ \Phi_1 \circ \mathcal{E} \right),
\end{align}
where $L$ denotes the depth of the network. The lifting operator $\mathcal{E}$ maps the physical inputs to a higher-dimensional latent feature $\bm{v}^{(0)}$. The operator, $\Phi_{l+1}$, transforms the hidden representation $\bm{v}^{(l)}$ to $\bm{v}^{(l+1)}$ via the sum of a local linear operator and a non-local integral kernel. The final representation $\bm{v}^{(L)}$ is mapped to the target solution $\bm{u} \in \mathcal{Y}$ via the projection operator $\mathcal{Q}$.

To generalize across the parameter space $\Theta$, the neural operator must adapt its integration geometry based on $\bm{\theta}$. In a standard neural operator, the integral kernel acts as a static mechanism, applying the same spatial aggregation rules across all problem instances. We overcome this limitation by replacing the static kernel with a parameter-gated mixture of kernels, allowing the network to dynamically assemble a custom integration rule for each specific PDE instance.

Let $\bm{W}$ be a learnable weight matrix performing a linear transformation, $\bm{\kappa}_k$ be $K$ distinct base kernel functions that define the non-local integral operator, $g_k$ be a hard gating function, and $\sigma$ be an activation function. $\bm{h}^{(l)}$ is the latent Dirichlet boundary condition created by the generalized boundary transfer operator as introduced in Section~\ref{sec:gbto}. We use the angle brackets $\langle \cdot, \cdot, \cdot \rangle$ to denote concatenation. 

To update the state at any specific target point $\bm{x}$, the integral simultaneously gathers information from surrounding spatial points (the physical domain) and from the governing conditions (boundaries and parameters). Physically, $\bm{\kappa}_k(\bm{x}, \bm{y})$ calculates the influence of the features at source point $\bm{y}$ on the target point $\bm{x}$. We project the PDE parameters into a higher-dimensional feature space, denoted by $\tilde{\bm{\theta}}$, via a learnable transformation. The update rule for the hidden state $\bm{v}^{(l)}(\bm{x})$ at layer $l$ is given by: 
\begin{equation} \label{eq:kernel}
\resizebox{0.9\textwidth}{!}{%
$
\displaystyle \bm{v}^{(l+1)}(\bm{x}) = \sigma\bigg( \bm{W} \langle \tilde{\bm{\theta}}, \bm{h}^{(l)}, \bm{v}^{(l)}(\bm{x}) \rangle + \int_{\Omega \cup \mathcal{X}_{ctx}} \left( \sum_{k=1}^{K} g_k(\bm{\theta}) \bm{\kappa}_k(\bm{x}, \bm{y}) \right) \bm{\xi}(\bm{y}) \, d \mu(\bm{y}) + \bm{b}(\bm{x}) \bigg), \quad \forall \bm{x} \in \Omega
$
}
\end{equation}
The measure $\mu(\bm{y})$ is defined as the standard Lebesgue measure for $\bm{y} \in \Omega$ and the counting measure for $\bm{y} \in \mathcal{X}_{ctx}$. The augmented input function $\bm{\xi}(\bm{y})$ represents the latent state $\bm{v}^{(l)}(\bm{y})$ for $\bm{y} \in \Omega$, and the conditioning information (parameter $\bm{\theta}$ or boundary $\bm{h}$) for $\bm{y} \in \mathcal{X}_{ctx}$. In practice, we adopt self-attention~\citep{vaswani2023attentionneed} kernels as this is the best approach reported in \citet{wang2025fdbenchmodularfairbenchmark}.

We employ the hard-gating strategy to select the most relevant kernel for a given parameter configuration. Let $s_k(\bm{\theta})$ be a learnable, lightweight score function that measures the affinity between the physical parameter $\bm{\theta}$ and the kernel $\kappa_k$. We compute the selection probabilities $\tilde{g}_k$ and the discrete gates $g_k$ as follows:
\begin{align}
    \tilde{g}_k(\bm{\theta}) = \frac{\exp(s_k(\bm{\theta}))}{\sum_{j=1}^{K} \exp(s_j(\bm{\theta}))}, \ \
    g_k(\bm{\theta}) = 
    \begin{cases} 
    1, & \text{if } k = \arg\max_j \tilde{g}_j(\bm{\theta}) \\ 
    0, & \text{otherwise} 
    \end{cases}
\end{align}
\textbf{Motivation:} Our parameter-gated mixture of kernels is motivated by both experimental observations and physical intuition. We initially tested strong single-kernel baselines that explicitly encode PDE parameters, specifically \textit{CAPE+FNO} and \textit{CAPE+Attn} from \citet{takamoto2023learningneuralpdesolvers}, as well as \textit{Unisolver}~\citep{zhou2025unisolver}. We observed that performance varies significantly depending on the PDE parameters, suggesting that single-kernel models often excel in specific domains while failing in others. This necessitates a mixture of kernels to better adapt to the diverse physical conditions created by varying PDE parameters. We refer readers to the Appendix~\ref{appendix:why_mixture} for full experimental details.

Physically, variations in PDE parameters often drive the system across distinct behavioral regimes rather than merely scaling the output magnitude. A classic example is the Reynolds number in fluid dynamics. Shifting this parameter transitions the system from laminar flow to turbulence. A single kernel forces a compromise, attempting to learn an average operator that is often too diffusive for turbulence yet too reactive for laminar states. Our parameter-gated approach mirrors this physical reality by treating different parameter spaces as distinct regimes.

Further, as emphasized in Section~\ref{sec:intro}, efficiency is paramount for the practical deployment of foundational PDE solvers. Our mixture of kernels incurs minimal computational overhead during inference, allowing the model to achieve superior accuracy while maintaining a significant speed advantage over traditional numerical solvers.

\textbf{Discussion with Related Work: }
Mixture of kernel-based neural operators have been explored in previous works. \citet{deighan2025mixtureneuraloperatorexperts} use all kernels where gating depends on the spatial coordinates. \citet{wang2025mixtureofexpertsoperatortransformerlargescale} use more than one kernel and use latent initial conditions as gating. By contrast, we only select one kernel to ensure efficiency and use PDE parameters for explicit domain partition. These innovations, though simple, are crucial for model generalization and practical utility.

\subsection{Generalized Boundary Transfer Operator}
\label{sec:gbto}
\textbf{Mathematical Motivation: }
The Dirichlet-to-Neumann (DtN) and Neumann-to-Dirichlet (NtD) operators~\citep{arendt2014dirichlet, sauter2013degenerate, bossavit1991scalar, 10.1093/oso/9780198501787.001.0001, oberai1998implementation, knockaert2008complex} act as inverse transfer functions for the domain. Let $\Omega_E$ be a bounded domain with a smooth boundary $\partial\Omega_E$, governed by a uniformly elliptic linear operator $\mathcal{L}_E$. The Dirichlet-to-Neumann (DtN) Operator, $\Lambda$, determines the necessary Neumann flux to sustain a given Dirichlet boundary state. Formally, $\Lambda(\bm{\varphi}) := \partial_{\bm{n}} \bm{u}_E|_{\partial\Omega_E}$, where $\bm{\varphi}$ is the Dirichlet boundary value. The Neumann-to-Dirichlet (NtD) Operator, $\eta$, operating in reverse, yields the Dirichlet boundary values that result from an applied Neumann flux. $\eta(\bm{\psi}) := \bm{u}_E|_{\partial\Omega_E}$, where $\bm{\psi}$ is a prescribed Neumann boundary flux.

\textbf{Proposed Method: } While foundational theory and the majority of existing literature on these operators focus primarily on time-independent elliptic PDEs with standard Neumann or Dirichlet conditions, we propose a more versatile approach. We propose to learn a generalized boundary transfer operator designed to map arbitrary boundary conditions to a shared latent Dirichlet representation. By unifying diverse boundary inputs into a single modality, this method allows the neural operator kernel to specialize in one type of boundary condition, thereby reducing complexity and enhancing model generalization. For time-dependent PDEs, strictly mapping a boundary condition to a Dirichlet output without the initial condition is an ill-posed problem. Thus, our transfer operator also incorporates the initial condition for capturing the system's complete boundary dynamics.

Let $\gamma: \mathcal{Y} \to \mathcal{Y}_{\partial\Omega}$ denote a trace operator, which restricts a function defined on the domain $\Omega$ to a latent value on the boundary $\partial\Omega$. We define this latent value as the latent Dirichlet value.
\begin{align}
    \bm{h} := \gamma(\bm{u}) := \gamma \left( \mathcal{S}_{\bm{\theta}, \mathcal{B}}(\bm{u}_0) \right),
\end{align}
where $\bm{u}$ is the solution trajectory of the PDE and $\bm{h}$ is the latent Dirichlet value. The trace operator $\gamma$ is a purely theoretical construct utilized solely to formally define the latent Dirichlet value. Because defining this value requires the ground truth trajectory, which is inherently unavailable during inference, we approximate it in practice using a generalized boundary transfer operator, $\mathcal{T}_{\bm{\lambda}}$, parametrized by $\bm{\lambda}$. We design $\mathcal{T}_{\bm{\lambda}}$ to function similarly to the classical DtN and NtD operators that map boundary data from one type to another. This operator maps the boundary configuration encapsulated by $\mathcal{B}$ and the interior dynamics to the equivalent latent Dirichlet data. We use the layer-wise latent state $\bm{v}^{(\ell)}$ as the input to the transfer operator to resolve the state-dependency of boundary interactions. In non-linear and time-dependent PDEs, the influence of a boundary condition is rarely static and is coupled with the instantaneous state of the field near the boundary. This allows the operator to dynamically adjust the boundary forcing based on the evolving physics of the interior domain. We argue that $\bm{v}^{(\ell)}$ serves as a more informative input than the initial condition $\bm{u}_0$ to the transfer operator.
\begin{align}
\bm{h}^{(l)} \approx \mathcal{T}_{\bm{\lambda}}(\mathcal{B}, \bm{v}^{(l)}).
\end{align}
This learned boundary data subsequently serves as the input for Equation~\ref{eq:kernel} to solve the PDE. In our implementation, $\mathcal{T}_{\bm{\lambda}}$ is parameterized as a light-weight, single-head self-attention, whose input consists of boundary type and associated boundary data.
% \YS{what is the input from $\mathcal{B}$? the pixels on the boundary? }

\subsection{Robust Optimization Objective}
In Appendix~\ref{appendix:why_mixture}, we observe that existing models exhibit large performance disparities across PDE parameters. Beyond physical domain partitioning, we attribute this issue to standard empirical risk minimization under MSE loss. A model can achieve low average loss by perfecting easy, smooth cases, while failing catastrophically on difficult, high-frequency cases. This is especially pronounced in CAPE~\citep{takamoto2023learningneuralpdesolvers} variants. To build a generalized neural solver, we require uniform accuracy across PDE parameters to ensure reliability. To tackle this, we propose a Group Distributionally Robust Optimization~\citep{kuhn2025distributionallyrobustoptimization, Gorissen_2015, Rahimian_2022, blanchet2024distributionallyrobustoptimizationrobust} training objective.
\begin{equation} \label{eq:dro}
    \mathcal{L}_{\text{DRO}}(\phi)
    =
    \max_{\bm{\theta} \in \bm{\Theta}'} \mathbb{E}_{(\bm{u}_0, \bm{u}, \mathcal{B}) \sim \mathcal{D}_{\bm{\theta}}} \left[ \ell(\mathcal{M}_{\phi}(\bm{u}_0, \bm{\theta}, \mathcal{B}), \bm{u}) \right],
\end{equation}
where $\bm{\Theta}'$ denotes the PDE parameter set of a training batch, $\mathcal{D}_{\bm{\theta}}$ denotes the subset of training data corresponding to the PDE parameter $\bm{\theta}$, and $\ell$ represents the MSE. Since the hard maximization in Equation~\ref{eq:dro} is non-differentiable, we approximate the objective using a differentiable Log-Sum-Exp~\citep{zhang2023dive} relaxation.
\begin{equation} \label{eq:soft_dro}
    \mathcal{L}_{\text{DRO}}(\phi) \approx 
    \frac{1}{\tau} \log \left( \frac{1}{|\bm{\Theta}'|} \sum_{\bm{\theta} \in \bm{\Theta}'} e^{ \left( \tau  \mathbb{E}_{(\bm{u}_0, \bm{u}, \mathcal{B}) \sim \mathcal{D}_{\bm{\theta}}} \left[ \ell(\mathcal{M}_{\phi}(\bm{u}_0, \bm{\theta}, \mathcal{B}), \bm{u}) \right] \right)} \right),
\end{equation}
where $\tau$ is the temperature parameter. When the cardinality of the parameter set $|\bm{\Theta}'|$ is prohibitively large, or the space is continuous, we employ a partitioning strategy. We discretize $\bm{\Theta}'$ into coarser bins, treating each bin as a distinct group for the DRO calculation.

Our training curriculum proceeds by iterating through each boundary condition type. In each step, we minimize the DRO objective defined in Equation~\ref{eq:soft_dro}. Detailed procedural steps are provided in Algorithm~\ref{algo:curriculum_dro}. We found that initiating the training directly with the DRO objective is overly restrictive. Thus, we adopt a two-stage curriculum. First, we train the network using standard MSE to establish a baseline approximation of the physics. After $P$ epochs, we switch to the Group DRO loss to enforce uniform accuracy, specifically targeting and refining performance in the most challenging regimes.

\section{Experiment}
\label{sec:exp}
First, we individually validate the model's generalization ability across PDE parameters and boundary conditions by fixing the other to a specific setting. This ensures that we evaluate the model's robustness with respect to only one changing factor at a time. Second, we conduct extensive testing specifically covering simultaneous PDE parameters and boundary conditions variations. We demonstrate that our proposed method achieves superior performance while maintaining computational efficiency against numerical solvers. 

To evaluate generalization across PDE parameters, we utilize four distinct partial differential equations: the Heat, Advection, and Burgers equations, as well as the incompressible Navier-Stokes equations. For the experiments concerning boundary condition generalization and joint pre-training, we focus specifically on the Heat and Advection equations. Detailed descriptions of these datasets and justifications for our benchmark selection are provided in Appendix~\ref{appendix:datasets} and Appendix~\ref{appendix:justify_dataset}. All models predict the next-frame solutions. We evaluate trajectories using a 10-step rollout (extended to 50 steps for the Incompressible Navier-Stokes). We utilize the normalized MSE as the evaluation metric: $\text{nMSE} = \frac{\| \hat{\bm{u}} - \bm{u} \|^2}{ \text{var}( \bm{u} )},$ where $\bm{u}$ and $\hat{\bm{u}}$ are the ground truth and predicted fields.

\begin{wraptable}{r}{0.7\textwidth}
\centering
\scriptsize{%

\setlength\tabcolsep{4pt}%
\renewcommand\arraystretch{1.4}%
\begin{tabular}{l || c c c c}
\hline\thickhline
\rowcolor{CadetBlue!20}
Method & Heat & Advection & Burgers & Inc. NS \\
\hline

% Baseline Rows
ViT-2 & $3.15 \times 10^{-2}$ & $8.95 \times 10^{-1}$ & $9.43 \times 10^{-2}$ & $1.10 \times 10^{-2}$ \\
ViT-5 & $1.49 \times 10^{-2}$ & $2.90 \times 10^{-1}$ & $1.25 \times 10^{-2}$ & $1.50 \times 10^{-2}$ \\
ViT-10 & $1.08 \times 10^{-2}$ & $1.40 \times 10^{-1}$ & $8.36 \times 10^{-2}$ & $6.37 \times 10^{-2}$ \\
Concat & $1.60 \times 10^{-2}$ & $9.48 \times 10^{-1}$ & $4.19 \times 10^{-2}$ & $7.68 \times 10^{-2}$ \\
CAPE+Unet & $3.05 \times 10^{-1}$ & $1.61 \times 10^{-1}$ & $1.30 \times 10^{-2}$ & $1.03 \times 10^{-1}$ \\
CAPE+FNO & $1.28 \times 10^{-2}$ & $1.69 \times 10^{-1}$ & $7.53 \times 10^{-3}$ & $6.34 \times 10^{-2}$ \\
CAPE+Attn & $1.52 \times 10^{-2}$ & $1.54 \times 10^{-1}$ & $1.39 \times 10^{-2}$ & $4.60 \times 10^{-2}$ \\
Unisolver & $6.89 \times 10^{-3}$ & $4.53 \times 10^{0}$ & $9.88 \times 10^{-3}$ & $1.68 \times 10^{-2}$ \\
MoE-POT & $6.84 \times 10^{-1}$ & $1.49 \times 10^{-1}$ & $5.29 \times 10^{-2}$ & $2.10 \times 10^{-2}$ \\

\hline

\rowcolor[HTML]{FFFFE0}
\textbf{Ours} & 
\bm{$4.20 \times 10^{-3}$} & 
\bm{$6.03 \times 10^{-2}$} & 
\bm{$2.82 \times 10^{-3}$} & 
\bm{$5.23 \times 10^{-5}$} \\

\hline

\end{tabular}}%
\caption{Generalization across PDE parameters. We evaluate the normalized MSE on the Heat, Advection, Burgers, and Incompressible Navier-Stokes equations. For each equation, we vary the PDE parameters across the samples in the dataset to test generalization.}
\label{tab:param_results}
\end{wraptable}

\subsection{Generalization across PDE Parameters}
We include \textit{ViT-{k}}~\citep{dosovitskiy2021imageworth16x16words}, which aligns with recent foundational models by relying on initial frames to implicitly infer PDE parameters. $k$ denotes the number of initial frames. We include \textit{Concat}, which follows \citet{subramanian2023towards} and concatenates PDE parameters with initial conditions. We adapt the CAPE mechanism~\citep{takamoto2023learningneuralpdesolvers} to create \textit{CAPE+Unet}, \textit{CAPE+FNO}, and \textit{CAPE+Attn}. Finally, we compare against \textit{Unisolver}~\citep{zhou2025unisolver} and \textit{MoE-POT}~\citep{wang2025mixtureofexpertsoperatortransformerlargescale}. We refer readers to the Appendix~\ref{appendix:implementation} for implementation details.

We present the results in Table~\ref{tab:param_results}. 
Our proposed model achieves superior performance, largely due to the explicit disentanglement in our architecture. By using PDE parameters to govern kernel selection, we decouple control from prediction. The most significant gains are observed in the incompressible Navier-Stokes equation. We hypothesize that the modest parameter variation in this dataset allows our kernels to act as specialized experts and finely resolve the physical dynamics.

% \textbf{Additional Baselines.} The problem corresponds to conditional generation, for which many standard approaches exist. We compare with these baselines in Appendix~\ref{appendix:param_baselines}.

% \textbf{Analysis on Number of Kernels.} We study the effects of the number of kernels on model performance in Appendix~\ref{appendix:num_kernels}.

% \textbf{Analysis on Robust Optimization Hyperparameters.} We study the hyperparameters in Appendix~\ref{appendix:robust_hyerparam}.

\begin{figure}[t]
\centering
\begin{minipage}[t]{0.48\textwidth}
\centering
\scriptsize
\setlength{\tabcolsep}{6pt}
\renewcommand\arraystretch{1.4}
\begin{tabular}{l || c c}
\hline \hline
\rowcolor{CadetBlue!20}
Method & Heat & Advection \\
\hline
ViT-2 & $2.33 \times 10^{-1}$ & $1.56 \times 10^{-1}$ \\
ViT-5 & $1.57 \times 10^{-1}$ & $1.48 \times 10^{-1}$ \\
ViT-10 & $9.18 \times 10^{-2}$ & $1.07 \times 10^{-1}$ \\
Unisolver & $1.15 \times 10^{-1}$ & $7.23 \times 10^{-2}$ \\
MoE-POT & $1.59 \times 10^{-1}$ & $4.14 \times 10^{-1}$ \\
\hline
\rowcolor[HTML]{FFFFE0}
\textbf{Ours} &
\bm{$7.80 \times 10^{-2}$} &
\bm{$5.59 \times 10^{-2}$} \\
\hline \hline
\end{tabular}
\captionof{table}{Generalization across boundary conditions, including Dirichlet, Neumann, and periodic. We report the normalized MSE averaging over different boundary conditions. For Dirichlet and Neumann boundaries, we randomly sample boundary values or fluxes.}
\label{tab:bc_generalization}
\end{minipage}
\hfill
\begin{minipage}[t]{0.48\textwidth}
\centering
\scriptsize
\setlength{\tabcolsep}{6pt}
\renewcommand\arraystretch{1.4}
\begin{tabular}{l || c c}
\hline \hline
\rowcolor{CadetBlue!20}
Method & Heat & Advection \\
\hline
ViT-2 & $2.60 \times 10^{1}$ & $3.33 \times 10^{-1}$ \\
ViT-5 & $2.44 \times 10^{0}$ & $2.51 \times 10^{-1}$ \\
ViT-10 & $2.06 \times 10^{0}$ & $2.28 \times 10^{-1}$ \\
Unisolver & $8.71 \times 10^{-1}$ & $5.02 \times 10^{-1}$ \\
MoE-POT & $1.85 \times 10^{0}$ & $3.80 \times 10^{-1}$ \\
\hline
\rowcolor[HTML]{FFFFE0}
\textbf{Ours} &
\bm{$3.62 \times 10^{-1}$} &
\bm{$7.97 \times 10^{-2}$} \\
\hline \hline
\end{tabular}
\captionof{table}{Generalization across PDE parameters and boundary conditions. We evaluate the normalized MSE on the Heat and Advection equations.}
\label{tab:main_results}
\end{minipage}
\end{figure}

\subsection{Generalization across Boundary Conditions}
Existing boundary enforcement strategies often lack generalization, being limited to specific boundary types~\citep{saad2023guiding} or elliptic PDEs~\citep{wang2024beno}. For comparison, we utilize two standard baselines, \textit{ViT-k}~\citep{dosovitskiy2021imageworth16x16words} and \textit{MoE-POT}~\citep{wang2025mixtureofexpertsoperatortransformerlargescale}, alongside the boundary-aware \textit{Unisolver}~\citep{zhou2025unisolver}.

Table~\ref{tab:bc_generalization} summarizes the results. We observe that increasing the number of initial frames enhances the performance of the \textit{ViT} variants. We attribute this to the temporal smoothness of the boundary data, which allows the model to implicitly extrapolate future boundary conditions from the extended history. Nevertheless, explicit boundary encoding remains superior, as demonstrated by our proposed model achieving the highest accuracy.

% \vspace{-2em}
\subsection{Unified Generalization: PDE Parameters and Boundary Conditions}
% We compare with \textit{ViT-{k}}~\citep{dosovitskiy2021imageworth16x16words}, \textit{Unisolver}~\citep{zhou2025unisolver}, and \textit{MoE-POT}~\citep{wang2025mixtureofexpertsoperatortransformerlargescale}.

\begin{wrapfigure}{r}{0.5\columnwidth}
    \centering

    \begin{subfigure}{0.48\linewidth}
        \centering
        \includegraphics[width=\linewidth]{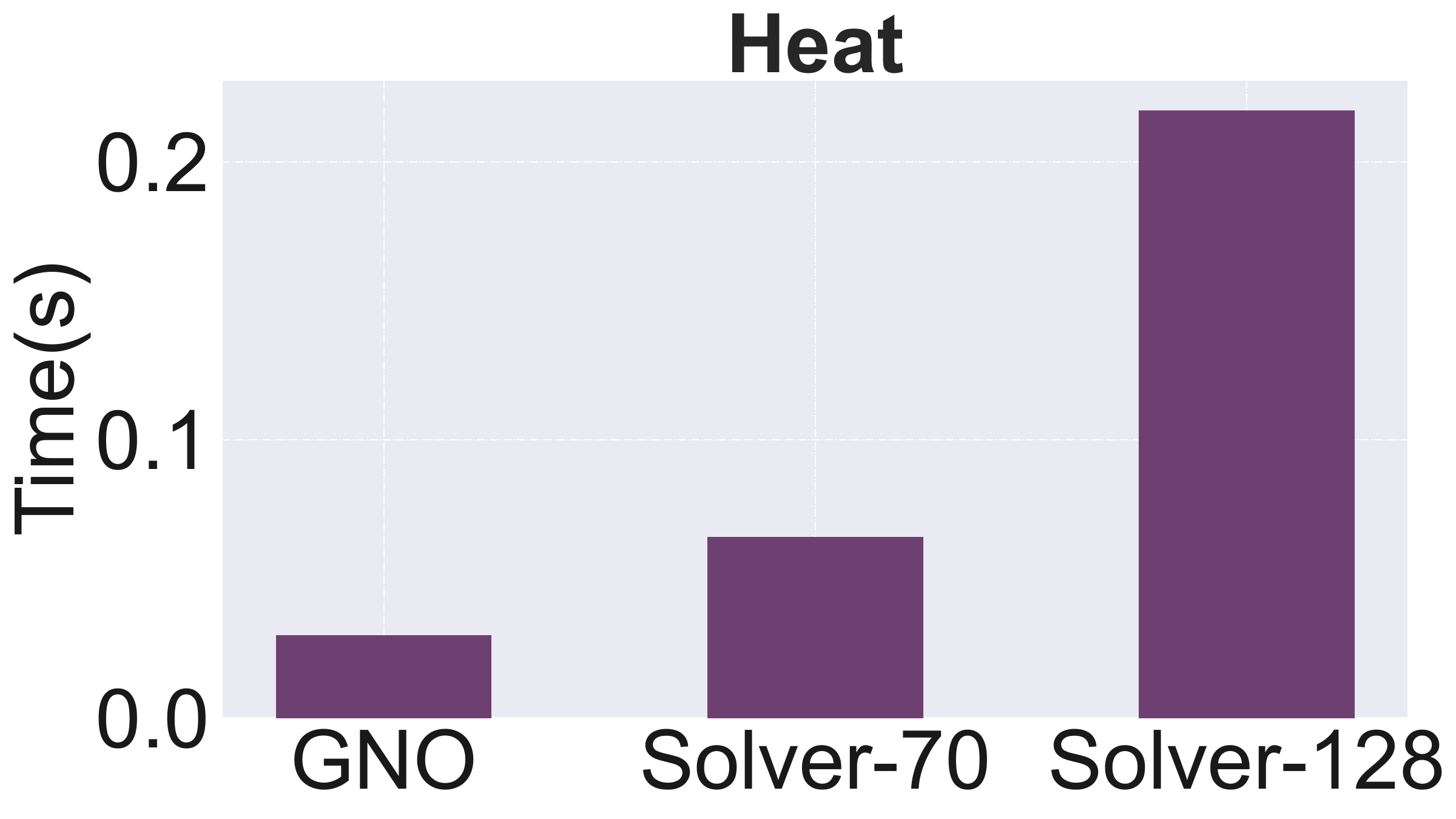}
        \caption{Heat}
    \end{subfigure}
    \hfill
    \begin{subfigure}{0.48\linewidth}
        \centering
        \includegraphics[width=\linewidth]{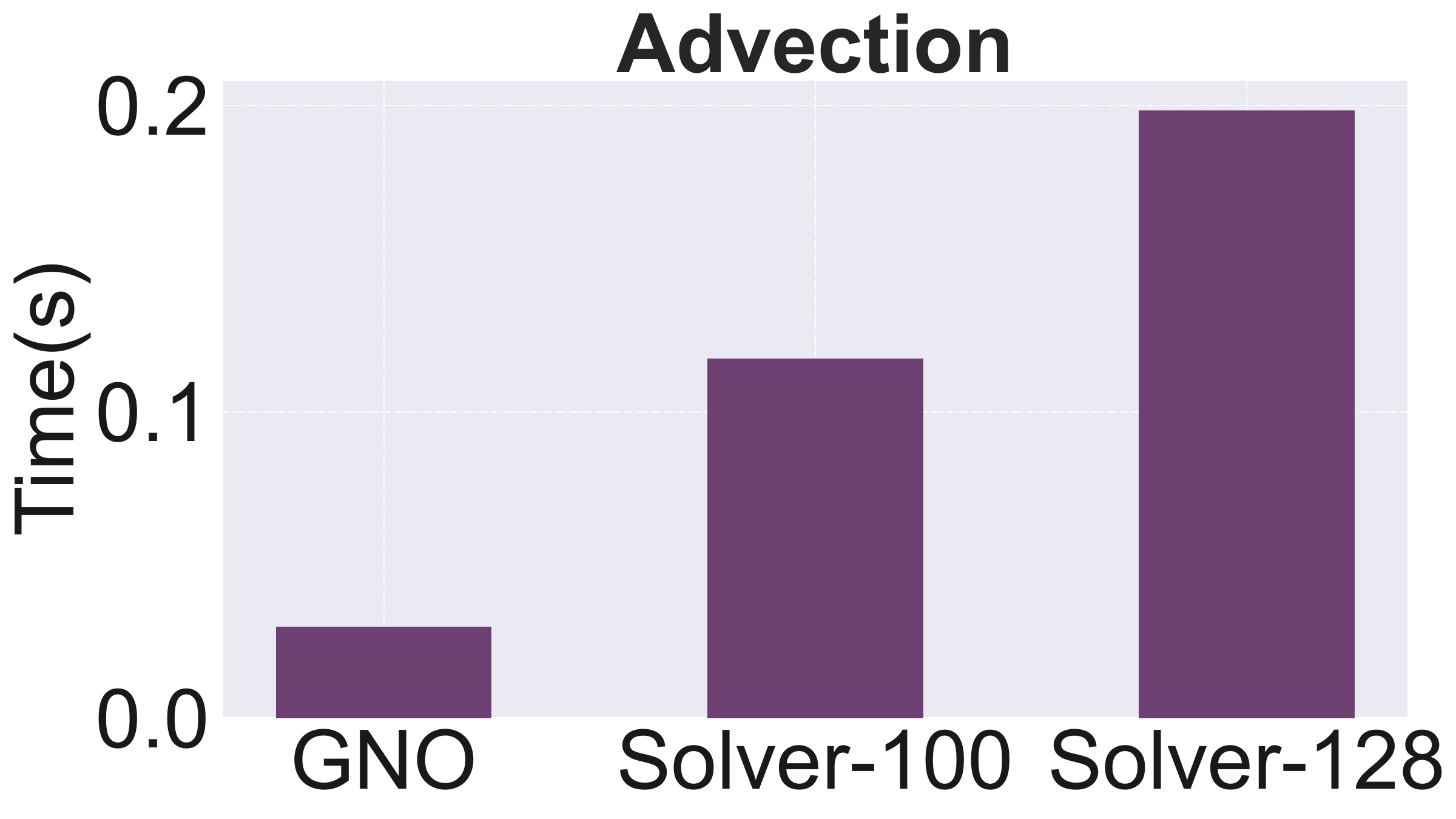}
        \caption{Advection}
    \end{subfigure}
    
    \caption{Runtime comparison between our \textit{GNO} and numerical solvers. The ground truth, operating at a $128 \times 128$ resolution, is denoted as Solver-128. For our accuracy-matched baselines, we utilize Solver-70 and Solver-100, which operate at $70 \times 70$ and $100 \times 100$ resolutions, respectively.}
    \label{fig:runtime}
\end{wrapfigure}

Table~\ref{tab:main_results} presents main results, demonstrating that our proposed model outperforms baselines. This performance not only highlights the critical role of explicit encoding of PDE parameters and boundary conditions, but also validates the superiority of our design choices. The PDE parameters in the Heat equation generate drastically different physical regimes. Combined with varying boundary conditions, this creates a highly complex learning landscape. While several baselines exhibit large nMSE values, suggesting a failure to learn, our method successfully adapts to these varying physical conditions.

Following \citet{McGreivy_2024, wang2025fdbenchmodularfairbenchmark}, we compare the runtime of our neural model against numerical solvers calibrated to lower resolutions that yield equivalent error. Specifically, we progressively reduce the spatial resolution of the numerical solvers until their accuracy matches that of the neural solver. This approach serves as a fair-comparison methodology to rigorously evaluate the trade-off between accuracy and computational efficiency when comparing numerical and neural solvers.

As shown in Figure~\ref{fig:runtime}, we observe an approximate $2\times$ speedup for the Heat equation and a $4\times$ speedup for the Advection equation. While this speedup is more modest than the orders-of-magnitude improvements reported in some prior literature, under such strict evaluation protocols, many existing neural solvers actually perform slower than numerical baselines. As noted in \citet{McGreivy_2024}, even the most competitive models achieve only a marginal $7\%$ speedup on datasets with fixed PDE parameters and boundary conditions. Consequently, our model demonstrates a highly favorable accuracy-efficiency trade-off.

\textbf{Performance Stability}
We report the standard deviation of the error on the Advection equation in Figure~\ref{fig:std} in Appendix~\ref{appendix:std}. Our model exhibits the lowest standard deviation, demonstrating that our training curriculum and robust optimization loss ensure reliable performance suitable for real-world applications. Interestingly, \textit{ViT-2} also shows low standard deviation. This is likely because the limited input context forces the model to predict averaged dynamics.

\subsection{Additional Experiments}

\textbf{Ablation Study and Model Analysis: } Appendix~\ref{appendix:ablation}

\textbf{Additional Experiment on 3D dataset: } Appendix~\ref{appendix:3d}

\textbf{Additional Experiment on Out-of-Distribution Generalization: } Appendix~\ref{appendix:ood}

\textbf{Additional Comparison with Physics Informed Neural Operator (PINO): } Appendix~\ref{appendix:pino}

\textbf{Additional Experiment on Complex PDE: } Appendix~\ref{appendix:complex}

\section{Conclusion}
We address the critical limitations of existing foundational models by proposing a theoretically grounded \textit{Generalized Neural Operator}. BWe achieve superior generalization while maintaining the computational efficiency of numerical solvers.

\newpage

\bibliography{neurips_2026}
\bibliographystyle{neurips_2026}

%%%%%%%%%%%%%%%%%%%%%%%%%%%%%%%%%%%%%%%%%%%%%%%%%%%%%%%%%%%%

\newpage

\appendix

\section{Broader Impact}
This work significantly advances the field of scientific machine learning by establishing a theoretically rigorous framework for neural PDE solvers. By demonstrating that efficient, parameter and boundary-aware architectures can outperform numerical solvers, we promote a more sustainable computational paradigm that reduces the energy consumption and environmental footprint associated with training large-scale networks. Furthermore, our emphasis on well-posedness enhances the reliability of AI-driven simulations in critical engineering applications, fostering greater trust in data-driven methods.

We do not anticipate any negative social or ethical consequences resulting from this work.

\section{Limitation and Future Work}
A primary limitation of this study lies in the scarcity of established benchmark datasets that incorporate diverse boundary conditions. Existing datasets in the domain predominantly focus on fixed or simplistic boundary scenarios, which restricts the comprehensive evaluation of generalized neural operators. While we have mitigated this issue by generating a diverse set of boundary variations to rigorously validate our approach, the lack of a standardized community benchmark remains a challenge. We hope that future work will address this gap by developing high-quality, unified datasets that systematically capture the complexity of varying boundary interactions in physical systems. Establishing such comprehensive resources will require a sustained, long-term community effort and cannot be easily resolved within the scope of our current work.

\section{Datasets}
\label{appendix:datasets}
\textbf{Heat Equation}
The heat equation models the diffusion of temperature in a domain 
$\Omega \subset \mathbb{R}^d$:
\[
\begin{cases}
\dfrac{\partial u(\bm{x},t)}{\partial t} = \alpha \nabla^2 u(\bm{x},t), & (\bm{x},t) \in \Omega \times (0,T], \\[6pt]
\mathcal{B}[u] = 0 & \\[6pt]
u(\bm{x},0) = u_0(\bm{x}), & \bm{x} \in \Omega,
\end{cases}
\]
where $u(\bm{x},t)$ is the temperature field, $\bm{x} \in \mathbb{R}^2$ denotes the spatial position, 
and $t$ is time. The coefficient $\alpha > 0$ is the thermal diffusivity, defined as $\alpha = \frac{k}{\rho c_p},$ where $k$ is the thermal conductivity, $\rho$ is the density, and $c_p$ is the specific heat capacity. Here, the PDE parameter is the thermal diffusivity $\alpha$.

To evaluate generalization across PDE parameters, we examined thermal diffusivity $\alpha \in [0.01, 1.0]$. For the training data, we defined a linearly spaced parameter grid $A_{\text{train}} = \{\alpha_i\}_{i=1}^{501}$ over this interval. At each $\alpha_i \in A_{\text{train}}$, we generated 10 distinct trajectories, with initial conditions sampled randomly from a Gaussian field under strictly periodic boundary conditions. For validation and testing, we evaluated the model's interpolation capabilities by sampling a denser grid $A_{\text{test}} = \{\tilde{\alpha}_j\}_{j=1}^{1001}$. For each $\tilde{\alpha}_j \in A_{\text{test}}$, we generated 2 trajectories.

Next, we assessed generalization across boundary conditions by fixing the thermal diffusivity and varying the constraints among Dirichlet, Neumann, and periodic types. For the Dirichlet case, we generated 100 setups with boundary values sampled uniformly from $[-10, 10]$. Similarly, for the Neumann condition, we created 100 setups with boundary fluxes sampled from the same range. For both types, we generated 10 trajectories per setup. Additionally, we generated 1,000 trajectories for the periodic boundary condition. The validation and test sets mirrored this structure but generated only 2 trajectories per specific boundary setup.

Finally, we conducted joint pre-training to handle simultaneous variations in parameters and boundary conditions. We utilized the same diffusivity range ($0.01$ to $1.0$) with 501 evenly spaced training values. For each diffusivity value, we created 5 distinct boundary setups for each type (Dirichlet, Neumann, and periodic), where values and fluxes were randomly sampled. We generated 10 trajectories for each unique parameter-boundary combination using Gaussian-sampled initial conditions.

\textbf{Advection Equation}
The 2D advection equation models the transport of a scalar quantity 
$u(x,t)$ (such as temperature, density, or concentration) within a flow field 
$\mathbf{v}(x,t) = (v_1, v_2)^\top$ over a spatial domain $\Omega \subset \mathbb{R}^2$:

\[
\left\{
\begin{aligned}
    &\frac{\partial u}{\partial t} + \bm{v}(\bm{x}) \cdot \nabla u = 0, 
    && (\bm{x}, t) \in \Omega \times (0, T], \\
    &\mathcal{B}[u] = 0, 
    &&  \\
    & u(\bm{x}, 0) = u_0(\bm{x}), 
    && \bm{x} \in \Omega,
\end{aligned}
\right.
\]

where $\bm{x} = (x_1, x_2) \in \mathbb{R}^2$ and $t$ denotes time. The velocity $\bm{v}(\bm{x})$ is the PDE parameter. 

To investigate generalization across PDE parameters, we focused on varying the velocity field $\bm{v}(\bm{x})$. We generated the training dataset by randomly sampling 500 distinct velocity fields from a Gaussian field. For each velocity configuration, we simulated 10 distinct trajectories with initial conditions also drawn from a random Gaussian field. All simulations in this phase utilized periodic boundary conditions. For validation and testing, we expanded the scope by sampling 1,000 new velocity fields. For each of these test cases, we generated 2 trajectories to assess the model's performance on unseen parameters.

Subsequently, we evaluated generalization across boundary conditions by fixing the velocity field $\bm{v}(\bm{x})$ and varying the boundary constraints among Dirichlet, Neumann, and periodic types. For the Dirichlet condition, we created 100 configurations with boundary values sampled uniformly from $[-10, 10]$. Similarly, for the Neumann condition, we established 100 setups with boundary fluxes sampled from the same range. In each of these bounded cases, we generated 10 trajectories. Additionally, we produced 1,000 trajectories for the periodic boundary condition. In all scenarios, initial conditions were sampled from a random Gaussian field. The validation and test datasets followed this structure, generating 2 trajectories for each specific boundary setup.

Finally, we undertook a joint pre-training approach to capture the coupled effects of varying parameters and boundary conditions. We sampled 500 velocity fields $\bm{v}(\bm{x})$ from a random Gaussian field for the training set. For every specific velocity field, we constructed 5 distinct boundary configurations for each type (Dirichlet, Neumann, and periodic), with random sampling applied to the Dirichlet values and Neumann fluxes. We then generated 10 trajectories for each unique combination of parameter and boundary settings. For validation and testing, we sampled 1,000 new velocity fields. Following the training protocol for boundary generation, we produced 2 trajectories per setup to evaluate the model's ability to generalize jointly across diverse physical parameters and boundary constraints.

\textbf{Burgers' Equation}
The 2D viscous Burgers' equation describes the evolution of a velocity field 
$\bm{u} = (u_1, u_2)^\top$ in a domain $\Omega \subset \mathbb{R}^2$:

\[
\left\{
\begin{aligned}
    &\frac{\partial \bm{u}}{\partial t} + (\bm{u} \cdot \nabla) \bm{u} = \nu \nabla^2 \bm{u}, 
    && (\bm{x}, t) \in \Omega \times (0, T], \\
    &\mathcal{B}[\bm{u}] = 0, 
    &&  \\
    & \bm{u}(\bm{x}, 0) = \bm{u}_0(\bm{x}), 
    && \bm{x} \in \Omega,
\end{aligned}
\right.
\]

where $\bm{u}(\bm{x},t) = (u_1(\bm{x},t), u_2(\bm{x},t))^\top$ is the velocity vector field, 
$\bm{x} = (x_1, x_2) \in \mathbb{R}^2$ is the spatial coordinate, 
and $t$ is time. The coefficient $\nu > 0$ is the kinematic viscosity, 
controlling the strength of diffusive effects relative to nonlinear advection. 

% \textbf{Generalization across PDE Parameter} 
We examined the kinematic viscosity $\nu$ across distinct training and testing regimes. For the training data, we defined a linearly spaced parameter grid $V_{\text{train}} = \{\nu_i\}_{i=1}^{501}$ over the interval $[5 \times 10^{-4}, 5 \times 10^{-2}]$. At each $\nu_i \in V_{\text{train}}$, we generated 10 distinct trajectories of length $T=100$ time steps. Initial conditions were sampled randomly from a Gaussian field, and all simulations strictly enforced periodic boundary conditions. For the validation and test datasets, we evaluated a shifted parameter space, defining a denser grid $V_{\text{test}} = \{\tilde{\nu}_j\}_{j=1}^{1001}$ over the interval $[0.01, 1.0]$. For each $\tilde{\nu}_j \in V_{\text{test}}$, we generated 2 trajectories of $100$ time steps using the same initial and boundary condition protocols.

\textbf{Incompressible Navier-Stokes Equation}

The 2D incompressible Navier-Stokes equations govern the motion of a viscous, incompressible fluid within a domain $\Omega \subset \mathbb{R}^2$:

\begin{equation*}
\scalebox{0.85}{% 
$
\left\{
\begin{aligned}
    &\frac{\partial \bm{u}}{\partial t} + (\bm{u} \cdot \nabla)\bm{u} = -\nabla p + \nu \nabla^2 \bm{u} + \bm{f}, 
    && (\bm{x}, t) \in \Omega \times (0, T], \\
    &\nabla \cdot \bm{u} = 0, 
    && (\bm{x}, t) \in \Omega \times (0, T], \\
    &\mathcal{B}[\bm{u}, p] = 0, 
    &&  \\
    &\bm{u}(\bm{x}, 0) = \bm{u}_0(\bm{x}), 
    && \bm{x} \in \Omega,
\end{aligned}
\right.
$
}
\end{equation*}

where $\bm{u}(\bm{x}, t) = (u_1(\bm{x}, t), u_2(\bm{x}, t))^\top$ is the velocity field, 
$p(\bm{x}, t)$ is the pressure field, $\nu=\frac{\rho u L}{Re} > 0$ is the kinematic viscosity, and 
$\bm{f}(\bm{x}, t)$ represents an external body force such as gravity. $\rho$ is the fluid density. $u$ is the velocity of the fluid field. $L$ is the characteristic length. $Re$ is the Reynold number.
The first equation expresses the conservation of momentum, while the second enforces 
the incompressibility condition $\nabla \cdot \bm{u} = 0$, ensuring volume preservation of the flow.

% \textbf{Generalization across PDE Parameter}
We examine the Reynolds number across a range of values from 500 to 1000, inclusive. To create our training data, we select 101 evenly spaced values within this range. For each specific kinematic viscosity, we generate 10 distinct trajectories, each with a length of 100 time steps. The initial conditions for these trajectories are sampled randomly from a Gaussian field. For the validation and test datasets, we again look at kinematic viscosity ranging from 500 to 1000. However, we sample this space more densely, selecting 501 evenly spaced values. For each of these values, we generate 2 trajectories, each with a length of 100 time steps, using initial conditions sampled from a random Gaussian field. All simulations use periodic boundary conditions.

\section{Justifications on Benchmark Dataset Selection}
\label{appendix:justify_dataset}
Our benchmark datasets include a variety of PDEs. For PDE parameters, we consider wide variations. For boundary conditions, \textcolor{red}{we consider both local and nonlocal boundary conditions, including Dirichlet, Neumann, and periodic}, which are commonly studied boundary conditions in existing works~\citep{wang2024beno, saad2023guiding, wang2025fdbenchmodularfairbenchmark, takamoto2024pdebenchextensivebenchmarkscientific, ohana2024well}.

Existing benchmarks and datasets~\citep{wang2025fdbenchmodularfairbenchmark, takamoto2024pdebenchextensivebenchmarkscientific, ohana2024well} generally lack diverse boundary variations. Thus, we believe that our datasets constitute a valuable contribution and represent a significant step toward building generalized solvers. We suggest that future work focus on providing more comprehensive datasets regarding boundary variations.

\section{Model Implementation Details}
\label{appendix:implementation}
\textbf{\textit{ViT-{k}}} We adopt the architecture design from \citet{dosovitskiy2021imageworth16x16words}, utilizing a latent size of 256 and a patch size of 16. The model comprises 8 attention blocks, each configured with 8 attention heads and an MLP ratio of 4.0. We concatenate the $k$ initial conditions through the channel dimension.

\textbf{\textit{Concat}} The PDE parameters are spatially broadcast and concatenated with the initial condition to form the input for a ViT-based architecture~\citep{dosovitskiy2021imageworth16x16words}. We employ a patch size of 16 and a latent dimension of 256. The model consists of 8 attention blocks, each featuring 8 attention heads and an MLP ratio of 4.0.

\textbf{\textit{CAPE+Unet}} We combine the CAPE~\citep{takamoto2023learningneuralpdesolvers} with Unet~\citep{ronneberger2015unetconvolutionalnetworksbiomedical}. The CAPE module is configured with a widening factor of 64, a kernel size of 5, and a normalized dimension of 128. Additionally, $\text{if\_11cnv}$ is enabled, and the number of parameter embedding channels is set to 3. For the U-Net architecture, we utilize 32 initial features. The CAPE module accepts a single initial condition frame along with the PDE parameter, expanding the channel dimension to generate latent initial conditions. This representation is then fed into the U-Net to predict the subsequent frame.

\textbf{\textit{CAPE+FNO}} We combine the CAPE~\citep{takamoto2023learningneuralpdesolvers} with FNO~\citep{li2021fourier}. The CAPE module is configured with a widening factor of 64, a kernel size of 5, and a normalized dimension of 128. Additionally, $\text{if\_11cnv}$ is enabled, and the number of parameter embedding channels is set to 3. For FNO, we use modes 12 and a width of 20. The CAPE module accepts a single initial condition frame along with the PDE parameter, expanding the channel dimension to generate latent initial conditions. This representation is then fed into the FNO to predict the subsequent frame.

\textbf{\textit{CAPE+Attn}} We combine the CAPE~\citep{takamoto2023learningneuralpdesolvers} with ViT~\citep{dosovitskiy2021imageworth16x16words}. The CAPE module is configured with a widening factor of 64, a kernel size of 5, and a normalized dimension of 128. Additionally, $\text{if\_11cnv}$ is enabled, and the number of parameter embedding channels is set to 3. For ViT, we use a latent dimension of 256, a patch size of 16, and 8 attention blocks featuring 8 heads and an MLP ratio of 4.0. The CAPE module accepts a single initial condition frame along with the PDE parameter, expanding the channel dimension to generate latent initial conditions. This representation is then fed into the FNO to predict the subsequent frame.

\textbf{\textit{Unisolver}} We adopt the architecture design from \citet{zhou2025unisolver}, utilizing a latent size of 256 and a patch size of 16. The model comprises 8 attention blocks, each configured with 8 attention heads and an MLP ratio of 4.0.

\textbf{\textit{MoE-POT}} We adopt the architecture proposed by \citet{wang2025mixtureofexpertsoperatortransformerlargescale}. The model is configured with a latent dimension of 256, a patch size of 16, and 8 attention blocks featuring 8 heads and an MLP ratio of 4.0. Additionally, each block incorporates 4 experts to choose from. We use the top 2 experts and a shared common expert.

\textbf{\textit{Ours}} Our model is configured with a latent dimension of 256, a patch size of 16, and 8 attention blocks featuring 8 heads and an MLP ratio of 4.0. Each attention block contains 4 kernels (we use 3 kernels for the incompressible Navier Stokes). The gating function is a lightweight MLP with hidden size 32.

\section{Why Mixture of Kernel?}
\label{appendix:why_mixture}
\textbf{Error Trends vs. PDE Parameters} We plot the three strong baselines, \textit{CAPE+FNO}, \textit{CAPE+Attn}, and \textit{UniSolver}, showing normalized MSE (nMSE) as a function of the PDE parameter on the Incompressible Navier–Stokes dataset. The results reveal that nMSE varies markedly across parameters. All models degrade at low Reynolds numbers, while CAPE+FNO and CAPE+Attn exhibit an approximately exponential decay in error as the Reynolds number increases, reaching extremely small errors in the high-Re regime. The variability is most pronounced for CAPE+FNO, where nMSE is roughly 0.4 at low Re but becomes nearly negligible at high Re. These patterns suggest that the single kernel employed in the CAPE framework effectively captures PDE dynamics only within certain parameter regimes and can fail substantially outside those regimes. In contrast to the CAPE variants, UniSolver exhibits a noticeably smoother nMSE curve across PDE parameters, with no large fluctuations. The trend is still roughly parabolic, however, indicating parameter-dependent variation. This supports the same conclusion as before. A single kernel can capture PDE dynamics well only within certain parameter regimes. The results are reported in Figure~\ref{fig:param_vs_nmse}.

\begin{figure}[htbp]
  \centering
  \begin{subfigure}{0.33\textwidth}
    \includegraphics[width=\linewidth]{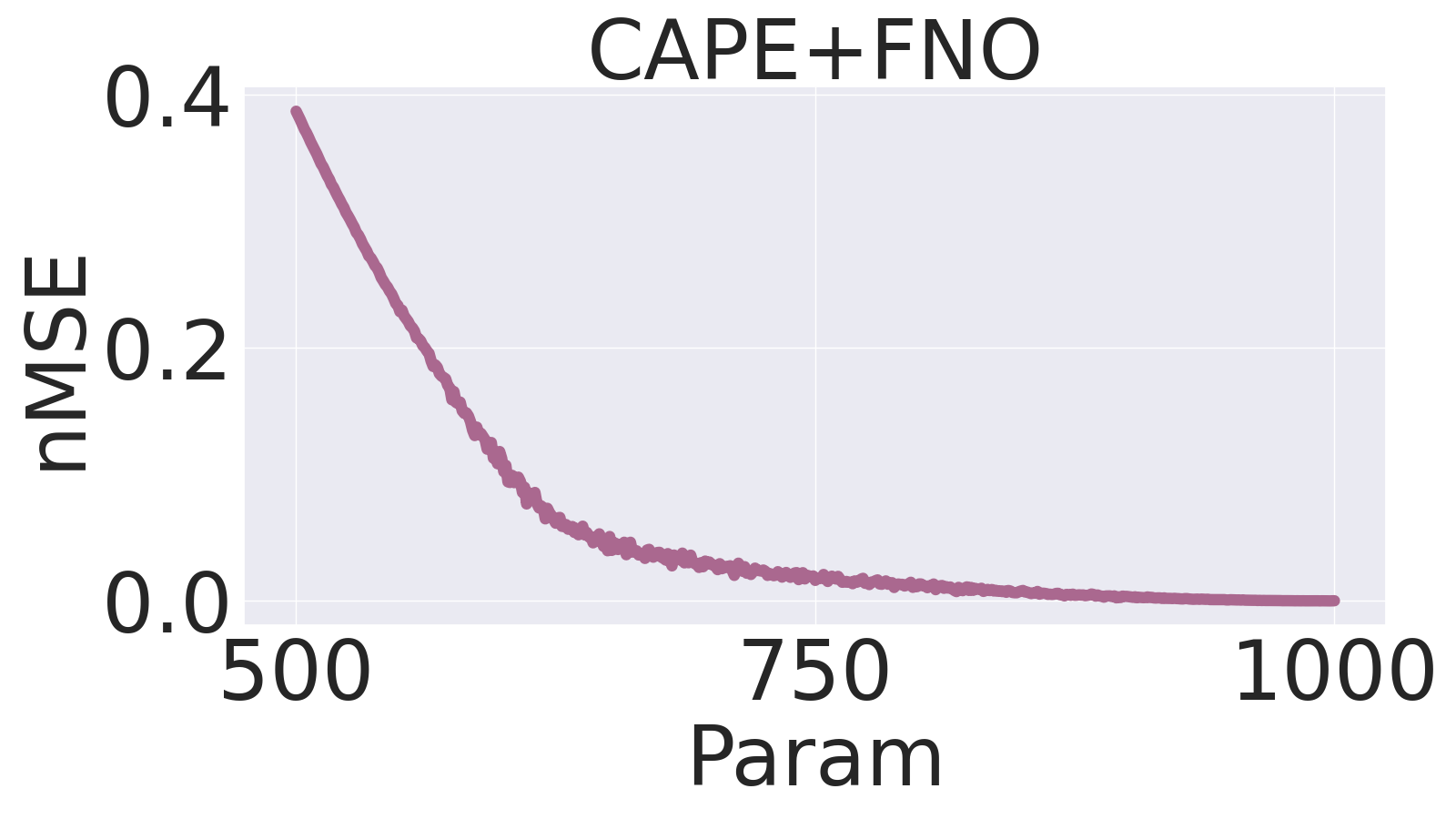}
  \end{subfigure}\hfill
  \begin{subfigure}{0.33\textwidth}
    \includegraphics[width=\linewidth]{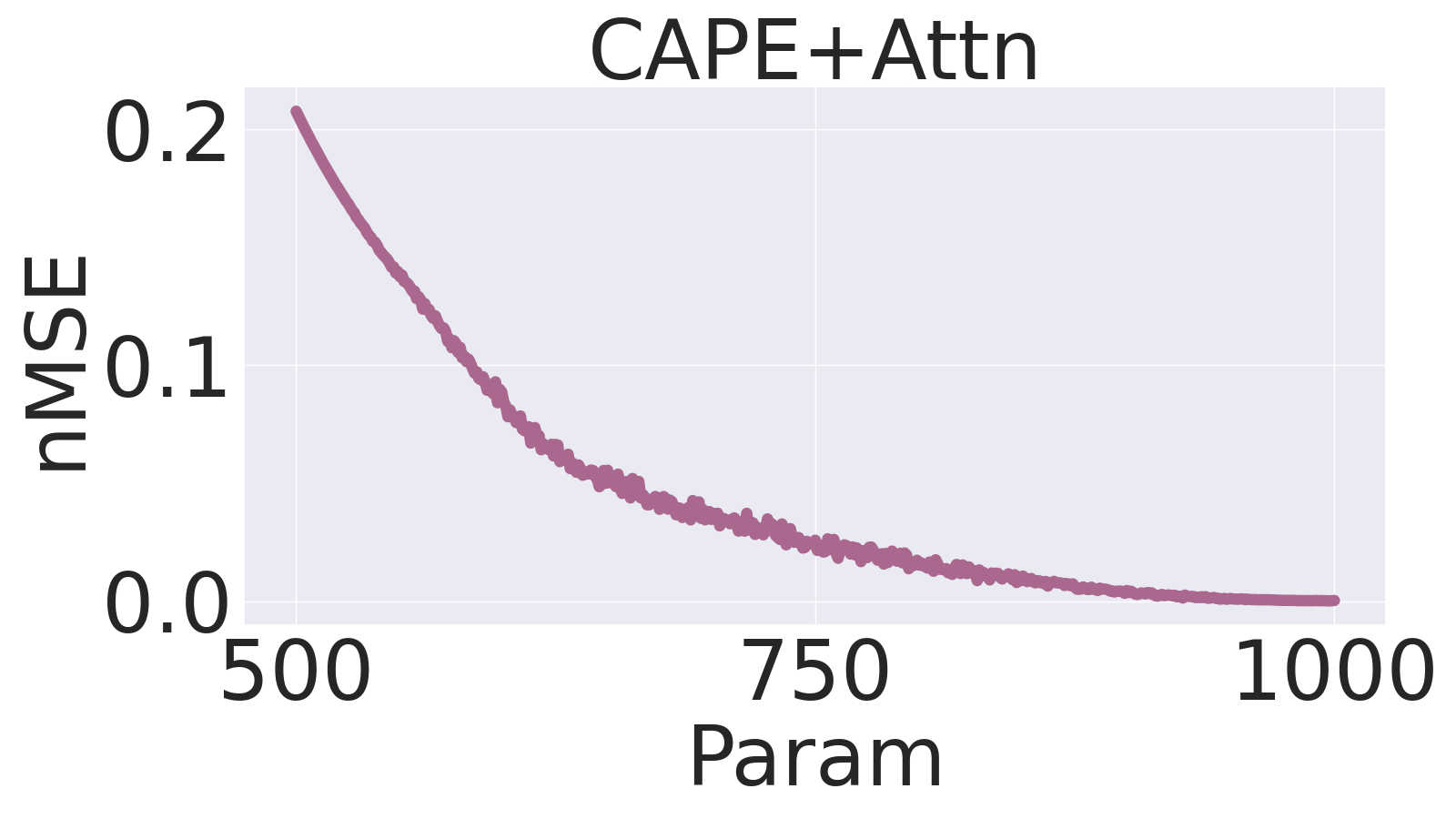}
  \end{subfigure}\hfill
  \begin{subfigure}{0.33\textwidth}
    \includegraphics[width=\linewidth]{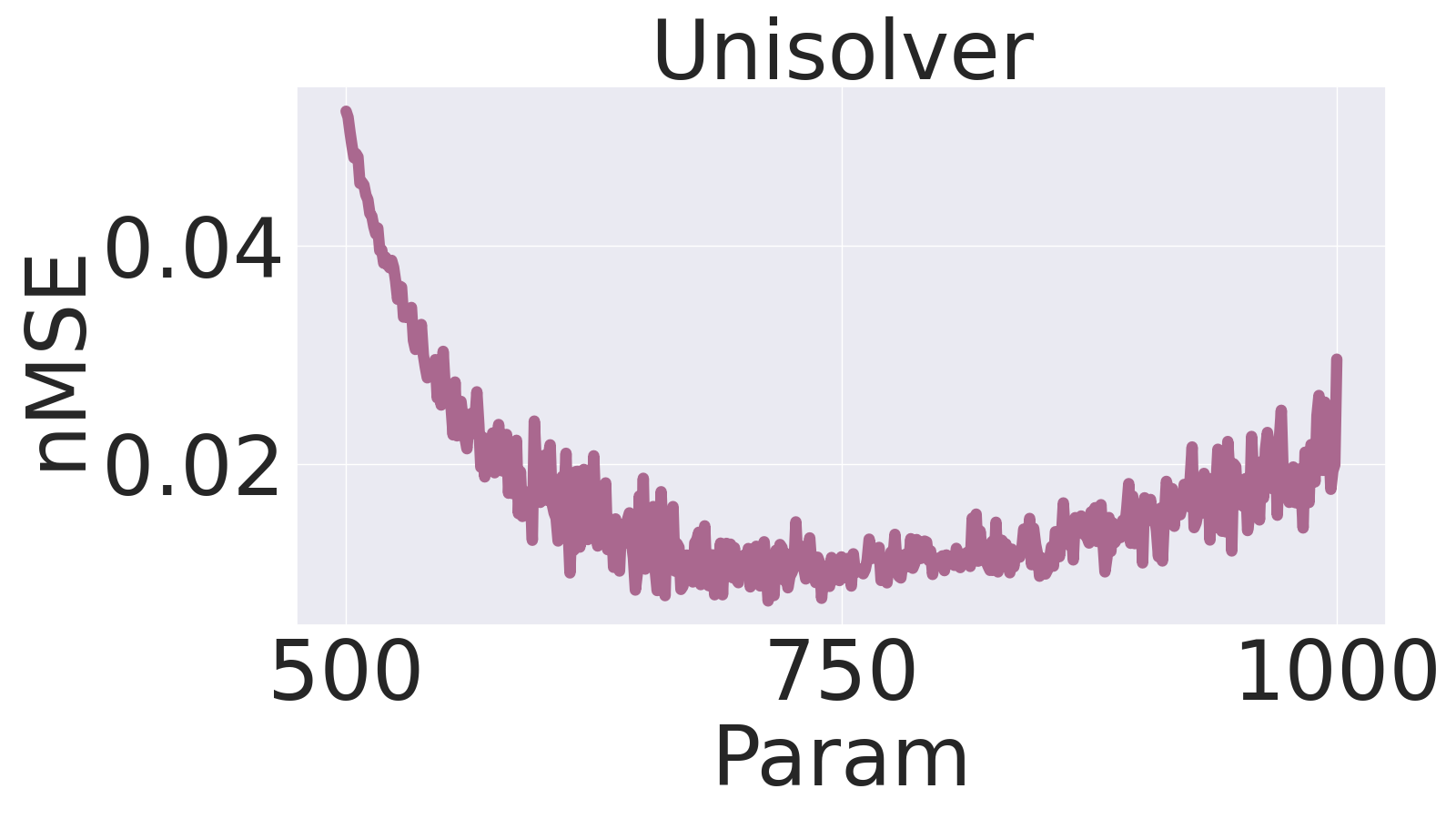}
  \end{subfigure}
  \caption{Error Trends vs. PDE Parameters on Incompressible Navier–Stokes dataset for \textit{CAPE+FNO}, \textit{CAPE+Attn}, and \textit{UniSolver}.}
  \label{fig:param_vs_nmse}
\end{figure}

\section{Additional Baselines for PDE Parameters Generalization}
\label{appendix:param_baselines}

\textbf{FiLM Encoder: } We apply FiLM~\citep{perez2017filmvisualreasoninggeneral} to the encoding stage of the model. $Z = \text{PatchEmbed}(x)$ and $\tilde{Z} = \gamma(\bm{\theta}) \odot Z + \beta(\bm{\theta})$, where $\gamma(\bm{\theta}) = W_{\gamma} \bm{\theta} + b_{\gamma}$ and $\beta(\bm{\theta}) = W_{\beta} \bm{\theta} + b_{\beta}$. Here $W_{\gamma}$, $b_{\gamma}$, $W_{\beta}$, and $b_{\beta}$ are learnable model parameters. $\bm{\theta}$ is the PDE parameter and $Z$ represent the latent embedding of initial condition.

\textbf{FiLM Attn: } We apply FiLM~\citep{perez2017filmvisualreasoninggeneral} before every multi-head attention layer. $\tilde{Z} = \gamma(\bm{\theta}) \odot Z + \beta(\bm{\theta})$ and $\tilde{Z} = \mathrm{MHA}\!\left(\mathrm{Query} = \tilde{Z},\ \mathrm{Key} = \tilde{Z},\ \mathrm{Value} = \tilde{Z} \right)$.

\textbf{Pos Encoder: } We apply parameter-guided positional encoding to the encoding stage of the model. $\tilde{Z} = Z + \text{PosEmbed}(\bm{\theta})$.

\textbf{Pos Attn: } We apply parameter-guided positional encoding before every multi-head attention layer. $\tilde{Z} = Z + \text{PosEmbed}(\bm{\theta})$ and $\tilde{Z} = \mathrm{MHA}\!\left(\mathrm{Query} = \tilde{Z},\ \mathrm{Key} = \tilde{Z},\ \mathrm{Value} = \tilde{Z} \right)$

\textbf{Proj Encoder: } We concatenate the input with the parameter and feed the result into an MLP in the encoding stage of the model. $\tilde{Z} = W [Z \mid \mid \bm{\theta}_{\text{expand}}] + b$. Here, $W$ and $b$ are learnable model parameters.

\textbf{Proj Attn: } We concatenate the input with the parameter and feed the result into an MLP before every multi-head attention layer. $\tilde{Z} = W [Z \mid \mid \bm{\theta}_{\text{expand}}] + b$ and $\tilde{Z} = \mathrm{MHA}\!\left(\mathrm{Query} = \tilde{Z},\ \mathrm{Key} = \tilde{Z},\ \mathrm{Value} = \tilde{Z} \right)$.

\textbf{Token Encoder: } We concatenate a parameter token at the encoding stage, and this token is not modified throughout the multi-head attention layers.

\textbf{Token Attn: } We concatenate a new parameter token before every multi-head attention layer. This token is only used during multi-head attention and is removed after this operation.

\begin{table}[h]
\centering
\scriptsize{%
% First argument (!) calculates width automatically
% Second argument (3.5cm) sets the fixed height
\resizebox{!}{2.4cm}{%
    \setlength\tabcolsep{6pt}%
    \renewcommand\arraystretch{1.1}%
\begin{tabular}{l  c}
\hline \hline
\rowcolor{CadetBlue!20}
Method & Inc. NS  \\
\hline

FiLM Encoder  & $1.85 \times 10^{-1}$ \\
FiLM Attn & $4.31 \times 10^{-2}$ \\
Pos Encoder & $2.75 \times 10^{-2}$ \\
Pos Attn & $1.22 \times 10^{-1}$ \\
Proj Encoder & Diverge \\
Proj Attn & $1.07 \times 10^{-2}$ \\
Token Encoder & $2.02 \times 10^{-1}$ \\
Token Attn & $3.69 \times 10^{-2}$ \\

\hline

\rowcolor[HTML]{FFFFE0}
\textbf{Ours} & 
\bm{$5.23 \times 10^{-5}$}  \\

\hline

\end{tabular}}}
\caption{Comparison of our proposed method against several additional baselines on conditional generation.}
\label{tab:additional_baselines}
\end{table}

\section{Additional Results on Model Stability}
\label{appendix:std}

We present the results in Figure~\ref{fig:std}.

\begin{figure}[htbp]
    \centering
    \includegraphics[width=0.49\columnwidth]{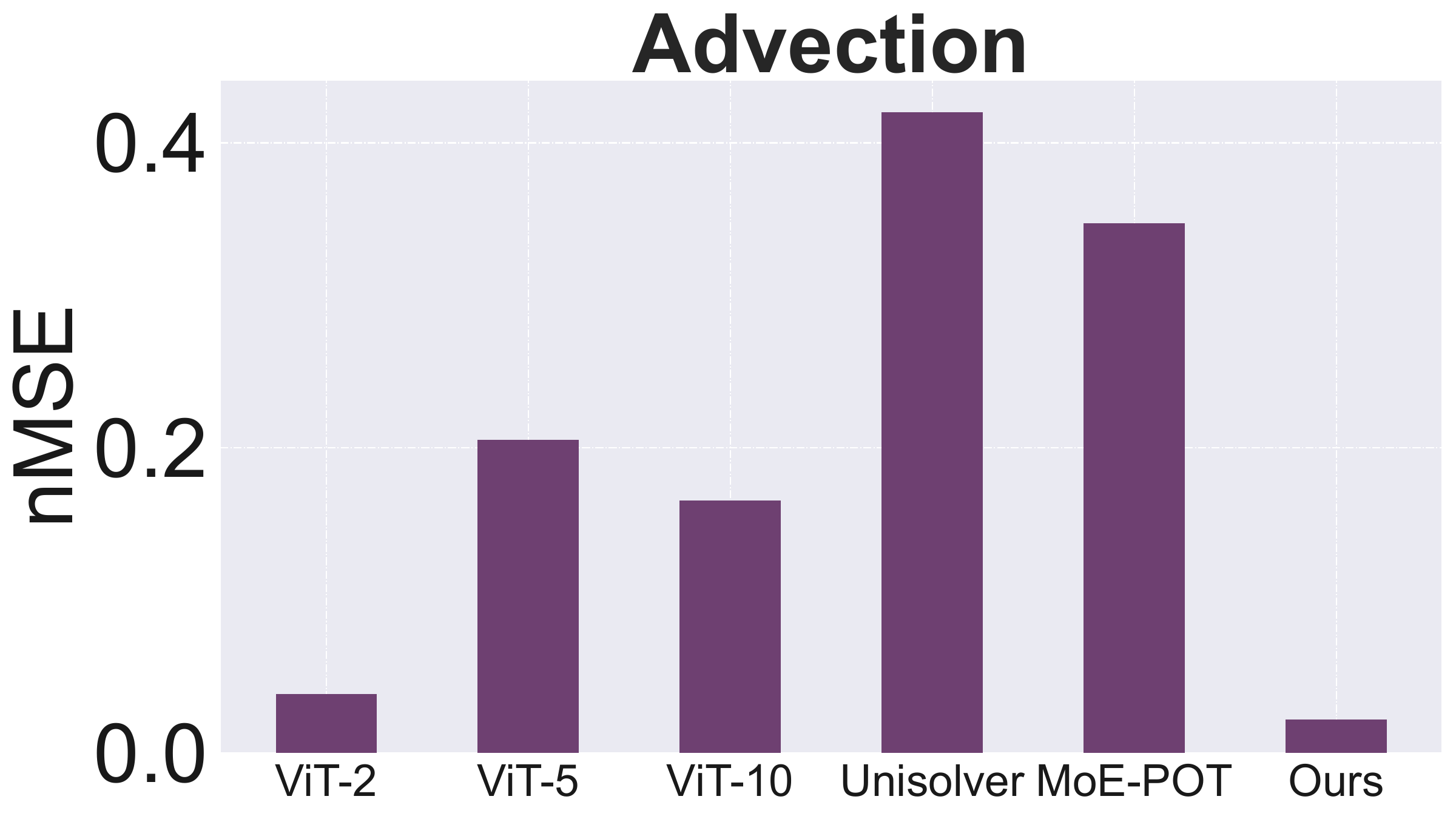}
    \caption{Standard deviation of model performance on the Advection equation in generalization across both PDE parameters and boundary conditions.}
    \label{fig:std}
\end{figure}

\section{Training Curriculum}
We present the training curriculum in Algorithm~\ref{algo:curriculum_dro}.

\begin{algorithm}
\caption{Training Curriculum with Group DRO}
\label{algo:curriculum_dro}
\begin{algorithmic}[1]
\Require Dataset $\mathcal{D}$ partitioned by boundary types and parameter groups $\bm{\Theta}'$; Total epochs $E$; Phase 1 limit $P$.

\For{$e \gets 1$ \textbf{to} $E$}
    \For{each boundary type}
        \State $\mathcal{D}' \gets \{( \bm{u}_0, \bm{u}, \bm{\theta}, \mathcal{B}) \in \mathcal{D} \mid \text{boundary type} \}$

        \If{$e \leq P$}
             \State \Comment{\textbf{Phase 1: Warm-up (Standard MSE)}}
             \State $\mathcal{L}_{\phi} \gets \mathbb{E}_{(\bm{u}_0, \bm{u}, \bm{\theta}, \mathcal{B}) \sim \mathcal{D}'} \big[ \ell(\mathcal{M}_{\phi}(\bm{u}_0, \bm{\theta}, \mathcal{B}), \bm{u}) \big]$
        \Else
            \State \Comment{\textbf{Phase 2: Robust Training (Group DRO)}}
            \State $L_{\bm{\theta}} \gets \mathbb{E}_{(\bm{u}_0, \bm{u}, \mathcal{B}) \sim \mathcal{D}'_{\bm{\theta}}} \big[\ell(\mathcal{M}_{\phi}(\bm{u}_0, \bm{\theta}, \mathcal{B}), \bm{u}) \big] \ \forall \bm{\theta}$
            \State $\mathcal{L}_{\phi} \gets \dfrac{1}{\tau} \log \left( \dfrac{1}{|\bm{\Theta}'|} \sum_{\bm{\theta} \in \bm{\Theta}'} \exp \left( \tau \, L_{\bm{\theta}} \right) \right)$
        \EndIf
        
        \State Update $\phi \gets \phi - \eta \nabla_{\phi} \mathcal{L}_{\phi}$
    \EndFor
\EndFor

\end{algorithmic}
\end{algorithm}

\section{Ablation Study and Model Analysis}
\label{appendix:ablation}

\subsection{Full Ablation Study}
\label{appendix:ablation}
Table~\ref{tab:ablation} presents the results of our ablation study. The top row, "Full Model," incorporates all proposed components: the parameter-gated mixture of kernels, the generalized boundary transfer operator, and the robust optimization objective. Removing any of these modules consistently hurts performance, validating their individual contributions. While the model retains PDE parameters even without the mixture of kernels, removing the generalized boundary transfer operator completely deprives the network of boundary information. This specific removal causes the most severe drop in performance, strongly supporting our argument that explicitly incorporating physical inputs is highly beneficial.

\begin{table}[h]
\centering
\scriptsize{
\resizebox{!}{1.1cm}{%
    \setlength\tabcolsep{6pt}%
    \renewcommand\arraystretch{1.1}%
\begin{tabular}{l  c}
\hline \hline
\rowcolor{CadetBlue!20}
Model & Advection  \\
\hline
\rowcolor[HTML]{FFFFE0} Full Model & \bm{$7.97 \times 10^{-2}$} \\
\thickhline
- Mixture of Kernels & $2.15 \times 10^{-1}$ \\
- Generalized Boundary Transfer Operator & $3.36 \times 10^{-1}$ \\
- Robust Optimization Objective & $9.22 \times 10^{-2}$ \\
\hline
\end{tabular}}}%
\caption{Ablation study on the Advection equation with simultaneous variations in PDE parameters and boundary conditions.}
\label{tab:ablation}
\end{table}

\subsection{Analysis on Number of Kernels}
\label{appendix:num_kernels}
We analyze the impact of kernel count on model performance, with results detailed in Table~\ref{tab:num_kernels}. Our findings confirm that increasing the number of kernels improves accuracy, validating the intuition that specialized kernels are necessary to handle distinct physical regimes arising from varied PDE parameters. Notably, the most significant performance gain occurs when increasing the number of kernels from 1 to 2. We hypothesize that because the parameter variations in the incompressible Navier-Stokes datasets are moderate, two kernels per block are sufficient to capture the diverse physical conditions. It is worth noting that since this mixture of kernels is present in every attention block, the increase in capacity scales with the depth of the model.

\begin{table}[h]
\centering
\scriptsize{
\resizebox{!}{1.5cm}{%
    \setlength\tabcolsep{6pt}%
    \renewcommand\arraystretch{1.1}%
\begin{tabular}{l  c}
\hline \hline
\rowcolor{CadetBlue!20}
\# Kernels & Inc. NS  \\
\hline

1  & $3.69 \times 10^{-2}$ \\
2 & $5.86 \times 10^{-5}$ \\
3 & $5.23 \times 10^{-5}$ \\
4 & $5.02 \times 10^{-5}$ \\
5 & $4.56 \times 10^{-5}$ \\

\hline

\end{tabular}}}%
\caption{Comparison of model performance with various kernel numbers.}
\label{tab:num_kernels}
\end{table}

\subsection{Analysis on Robust Optimization Objective Hyperparameters}
\label{appendix:robust_hyerparam}
We investigate the sensitivity of Algorithm~\ref{algo:curriculum_dro} to the switching epoch $P$ (the point at which the training objective transitions from standard MSE to robust optimization). Results are detailed in Tables~\ref{tab:p}. We find that while the model requires an initial MSE warm-up to prevent divergence, switching to the robust objective too late (e.g., epoch 90) degrades performance by limiting the time available for robust optimization. Consequently, we select $P=70$ to balance training stability with sufficient optimization time.

\begin{table}[h]
\centering
\scriptsize{
\resizebox{!}{1.7cm}{%
    \setlength\tabcolsep{6pt}%
    \renewcommand\arraystretch{1.1}%
\begin{tabular}{l  c}
\hline \hline
\rowcolor{CadetBlue!20}
$P$ & Inc. NS  \\
\hline
0  & Diverge \\
10  & $5.88 \times 10^{-5}$ \\
30 & $5.87 \times 10^{-5}$ \\
50 & $6.04 \times 10^{-5}$ \\
70 & $5.23 \times 10^{-5}$ \\
90 & $6.04 \times 10^{-5}$ \\

\hline

\end{tabular}}}%
\caption{Comparison of model performance on hyperparameter $P$ on the incompressible Navier-Stokes dataset.}
\label{tab:p}
\end{table}

\section{Additional Experiment on 3D dataset}
\label{appendix:3d}
Solving 3D PDEs with neural solvers is particularly difficult because the transition from 2D to 3D triggers a cubic explosion in data complexity. Beyond managing raw scale, the model must also capture intricate, non-local spatial interactions. To validate our model's performance in the 3D domain, we generated a 3D advection equation dataset on a $64^3$ grid, incorporating variations in PDE parameters and boundary conditions. As reported in Table~\ref{tab:3d}, purely data-driven methods, such as the ViT variants and MoE-POT, struggle significantly on this dataset. Specifically, the ViT variants require more initial condition frames to sufficiently infer the underlying physics. In contrast, our model achieves superior performance, outperforming baselines both with and without explicit physical inputs.

\begin{table}[ht]
\centering
\small
\setlength{\tabcolsep}{14pt}
\renewcommand\arraystretch{1.4}

\begin{tabular}{l || c}
\hline \hline
\rowcolor{CadetBlue!20}
Method & Advection \\
\hline

ViT-2 & $1.10 \times 10^{0}$ \\
ViT-5 & $4.22 \times 10^{-1}$ \\
ViT-10 & $4.16 \times 10^{-1}$ \\
Unisolver & $1.67 \times 10^{0}$ \\
MoE-POT & $8.48 \times 10^{-1}$ \\
\hline
\rowcolor[HTML]{FFFFE0}
\textbf{Ours} & \bm{$3.37 \times 10^{-1}$} \\
\hline \hline
\end{tabular}
\caption{Performance comparison on 3D Advection equation with both variations in PDE parameters and boundary conditions. We report the normalized MSE.}
\label{tab:3d}
\end{table}

\section{Additional Experiment on Out-of-Distribution Generalization}
\label{appendix:ood}
We evaluate our model's out-of-distribution (OOD) generalization using the advection dataset, introducing variations in both PDE parameters and boundary conditions. To test parameter shifts, we first calculate the Frobenius norm of the velocity fields in the training set. We then sample testing velocity fields with norms strictly outside the training range, specifically, values either smaller than the training minimum or larger than the training maximum.

For boundary conditions, the Dirichlet values and Neumann fluxes are sampled from [-10, 10] during training. In our OOD evaluation, these values are instead drawn from either the [-15, -10] or the [10, 15] intervals. This ensures that the model is tested on magnitudes entirely unseen during the training phase.

The results, summarized in Table~\ref{tab:ood}, demonstrate that our proposed model achieves optimal performance, highlighting its superior generalization capabilities compared to existing baselines. 

While the velocity fields for the original testing and validation data in Table~\ref{tab:main_results} are also unseen during training (as they are randomly sampled from a Gaussian field), their Frobenius norms may still fall within the range observed during training.

\begin{table}[ht]
\centering
\small
\setlength{\tabcolsep}{14pt}
\renewcommand\arraystretch{1.4}

\begin{tabular}{l || c}
\hline \hline
\rowcolor{CadetBlue!20}
Method & Advection \\
\hline

ViT-2 & $4.48 \times 10^{-1}$ \\
ViT-5 & $2.78 \times 10^{-1}$ \\
ViT-10 & $3.45 \times 10^{-1}$ \\
Unisolver & $5.35 \times 10^{-1}$ \\
MoE-POT & $4.24 \times 10^{-1}$ \\
\hline
\rowcolor[HTML]{FFFFE0}
\textbf{Ours} & \bm{$8.30 \times 10^{-2}$} \\
\hline \hline
\end{tabular}
\caption{Out-of-distribution comparison on the Advection equation with both variations in PDE parameters and boundary conditions. We report the normalized MSE.}
\label{tab:ood}
\end{table}

\section{Additional Comparison with Physics Informed Neural Operator (PINO)}
\label{appendix:pino}
The Physics-Informed Neural Operator (PINO) \citep{li2023physicsinformedneuraloperatorlearning} combines data-driven neural operators with physics-informed neural networks, leveraging both empirical data and symbolic equations via loss functions during training. However, PINO inherits a key limitation of traditional PINNs. It is tied to specific PDE parameters and boundary conditions. Because it does not explicitly receive physical parameters as inputs during inference, it struggles to generalize to new scenarios without retraining. To evaluate this, we trained PINO on the advection equation, using variations in both PDE parameters and boundary conditions purely to guide the training loss. As shown in Table~\ref{tab:pino}, our proposed model consistently outperforms PINO.

\begin{table}[ht]
\centering
\small
\setlength{\tabcolsep}{14pt}
\renewcommand\arraystretch{1.4}

\begin{tabular}{l || c c}
\hline \hline
\rowcolor{CadetBlue!20}
Method & Heat & Advection \\
\hline
PINO & $2.98 \times 10^{0}$ & $3.53 \times 10^{-1}$ \\ 
\hline
\rowcolor[HTML]{FFFFE0}
\textbf{Ours} & 
\bm{$3.62 \times 10^{-1}$} & 
\bm{$7.97 \times 10^{-2}$} \\
\hline \hline
\end{tabular}

\caption{Performance comparison between PINO and our proposed model. We report the normalized MSE on the Heat and Advection equations with both PDE parameter and boundary condition variations.}
\label{tab:pino}
\end{table}

\section{Additional Experiment on Complex PDE}
\label{appendix:complex}
Current benchmark datasets, such as those by \citet{ohana2024well, takamoto2024pdebenchextensivebenchmarkscientific, toshev2024lagrangebench}, typically lack simultaneous variations in both PDE parameters and boundary conditions for a single PDE type. Consequently, the primary datasets used in this study were self-generated. To further evaluate our model’s robustness, we also employ the 2D Turbulent Radiative Layer dataset from \citet{ohana2024well}. This simulation features a cold, dense gas layer beneath a hot, dilute gas layer, moving at highly subsonic relative velocities. The setup utilizes periodic boundary conditions in the $x$-direction and zero-gradient Neumann conditions in the $y$-direction. The governing PDE parameter, the cooling time, is varied across the set $\{0.03, 0.06, 0.1, 0.18, 0.32, 0.56, 1.00, 1.78, 3.16\}$. We model the density, pressure, and $x,y$ velocity fields. This validates our model on a public, computationally challenging PDE. We note that the boundary conditions remain fixed while only the PDE parameters vary.

We report the MSE in Table~\ref{tab:complex}. Our model achieves a considerably lower MSE compared to the baselines. Notably, purely data-driven ViT models perform poorly in this configuration. We hypothesize that because this PDE involves multiple coupled fields, ViT-based architectures struggle to capture the underlying physics from concatenated data. This highlights a significant limitation of purely data-driven approaches and underscores the superiority of our framework.

\begin{table}[ht]
\centering
\small
\setlength{\tabcolsep}{14pt}
\renewcommand\arraystretch{1.4}

\begin{tabular}{l || c}
\hline \hline
\rowcolor{CadetBlue!20}
Method & Advection \\
\hline

ViT-2 & $543.35$ \\
ViT-5 & $314.46$ \\
ViT-10 & $281.18$ \\
Unisolver & $35.43$ \\
MoE-POT & $24.82$ \\
\hline
\rowcolor[HTML]{FFFFE0}
\textbf{Ours} & \bm{$11.99$} \\
\hline \hline
\end{tabular}
\caption{Performance Comparison on 2D Turbulent Radiative Layer dataset. This is a public dataset containing complex PDE with variations only in PDE parameters. We report the MSE over 10 steps of rollout.}
\label{tab:complex}
\end{table}

\section{Proof of Proposition~\ref{prop:well}}
\label{appendix:proof_well}
We begin by outlining and justifying our assumptions.

\begin{assumption}[Existence]
\label{assump:existence}
    For every initial condition $\bm{u}_0 \in \mathcal{X}_{\text{init}}$, a solution $\bm{u} \in \mathcal{Y}$ to the problem defined in Equation~\ref{eq:problem} \textbf{exists}.
\end{assumption}
\begin{proof}[Justification]
This assumption is implicit in standard operator learning frameworks~\citep{kovachki2023neuraloperator}, which model the solution operator as a continuous map. The formulation of such a map inherently presumes that for every input initial condition, a realizable output exists.
\end{proof}

\begin{assumption}[Uniqueness]
\label{assump:uniqueness}
The solution to the problem defined in Equation~\ref{eq:problem} is locally \textbf{unique}. Specifically, for a fixed initial condition $\bm{u}_0$, there exists exactly one solution trajectory $\bm{u}$. This ensures that $\mathcal{S}_{\bm{\theta}, \mathcal{B}}$ is a single-valued mapping rather than a one-to-many relation.
\end{assumption}
\begin{proof}[Justification]
Uniqueness is a necessary condition for standard, deterministic neural operator architectures~\citep{kovachki2023neuraloperator}, which are designed to approximate single-valued functions. If the PDE admits non-unique solutions, they will fail to converge to a valid physical solution, often learning an unphysical average of the possible modes. Modeling one-to-many mappings requires probabilistic frameworks, such as diffusion-based neural operators~\citep{huang2024diffusionpdegenerativepdesolvingpartial, li2025selfguideddiffusionmodelaccelerating, li2025videopde, baldan2025flowmatchingmeetspdes}.
\end{proof}

\begin{assumption}[Non-Degeneracy of PDE Configuration]
\label{assump:identifiability}
The family of PDE problems defined by the configuration space $\mathcal{C} = \Theta \times \mathbb{B}$ is \textbf{non-degenerate}. We assume that the PDE parameters and boundary conditions exert a non-trivial influence on the system dynamics. If an initial condition $\bm{u}_0$ sufficiently excites the system, then there exist at least two distinct configurations $c_1, c_2 \in \mathcal{C}$ whose corresponding solution trajectories are not identical:
\begin{align}
    \exists c_1, c_2 \in \mathcal{C}, \quad c_1 \neq c_2 \quad \text{such that} \quad \mathcal{S}_{c_1}(\bm{u}_0) \neq \mathcal{S}_{c_2}(\bm{u}_0).
\end{align}
\end{assumption}
\begin{proof}[Justification]
This assumption reflects the principle of model parsimony and physical relevance. It postulates that every component of the parameter vector $\bm{\theta}$ and every constraint in the boundary operator $\mathcal{B}$ exerts a measurable influence on the system dynamics. If a specific parameter or boundary condition could be altered without affecting the solution trajectory $\bm{u}$, that element would be physically non-descriptive and effectively irrelevant to the process. We therefore assume the problem is formulated such that no such redundant setup exists.
\end{proof}

We proceed by contradiction. Suppose the generalized neural operator $\mathcal{M}_{\phi}$ is well-defined as a function dependent only on the initial condition $\bm{u}_0$, and that it correctly represents the family $\mathfrak{S}$.

This hypothesis implies that for any valid PDE configuration $c \in \mathcal{C}$ and any initial condition $\bm{u}_0$, the operator satisfies:
\begin{align}
    \mathcal{M}_{\phi}(\bm{u}_0) = \mathcal{S}_c(\bm{u}_0).
\end{align}

Consider the fixed initial condition $\bm{u}_0$ satisfying Assumption~\ref{assump:identifiability}. By the assumption, there exist distinct configurations $c_1, c_2 \in \mathcal{C}$ such that $\mathcal{S}_{c_1}(\bm{u}_0) \neq \mathcal{S}_{c_2}(\bm{u}_0)$.

However, if we apply our hypothesized operator $\mathcal{M}_{\phi}$ to the initial condition $\bm{u}_0$, it must yield the correct solution for both configurations:
\begin{align}
    \mathcal{M}_{\phi}(\bm{u}_0) &= \mathcal{S}_{c_1}(\bm{u}_0) \\
    \mathcal{M}_{\phi}(\bm{u}_0) &= \mathcal{S}_{c_2}(\bm{u}_0)
\end{align}
    
By definition, an operator is single-valued; it maps the input $\bm{u}_0$ to a unique output $\bm{y}^* \in \mathcal{Y}$. Substituting this into the equations above, we obtain:
    
\begin{align}
    \mathcal{S}_{c_1}(\bm{u}_0) = \bm{y}^* = \mathcal{S}_{c_2}(\bm{u}_0).
\end{align}

This implies $\mathcal{S}_{c_1}(\bm{u}_0) = \mathcal{S}_{c_2}(\bm{u}_0)$, which directly contradicts Assumption~\ref{assump:identifiability}. Thus, $\mathcal{M}_{\phi}$ cannot be solely a function of $\bm{u}_0$.

\section{Hardware Specification}
We implement all models in PyTorch~\citep{paszke2019pytorchimperativestylehighperformance}. All experiments are run on servers/workstations with the following configuration:

\begin{itemize}
    \item 80 CPUs, 503G Mem, 8 x NVIDIA V100 GPUs.
    \item 48 CPUs, 220G Mem, 8 x NVIDIA TITAN XP GPUs.
    \item 96 CPUs, 1.0T Mem, 8 x NVIDIA A100 GPUs.
    \item 64 CPUs, 1.0T Mem, 8 x NVIDIA RTX A6000 GPUs.
    \item 224 CPUs, 1.5T Mem, 8 x NVIDIA L40S GPUs.
    \item 128 CPUs, 480G Mem, 8 × NVIDIA RTX 4090 GPUs.
\end{itemize}

%%%%%%%%%%%%%%%%%%%%%%%%%%%%%%%%%%%%%%%%%%%%%%%%%%%%%%%%%%%%

% \newpage
% \input{checklist.tex}

\end{document}